\documentclass[a4paper,11pt]{quantumarticle}
\pdfoutput=1
\usepackage[utf8]{inputenc}
\usepackage[english]{babel}
\usepackage[T1]{fontenc}
\usepackage{amsmath,amsthm}
\usepackage{hyperref}
\usepackage{braket}
\usepackage{appendix}
\usepackage{enumitem}

\usepackage{tikz}
\usepackage{lipsum}
\usepackage{adjustbox}

\usepackage[sort&compress,numbers]{natbib}

\newtheorem{theorem}{Theorem}
\newtheorem{lemma}{Lemma}

\theoremstyle{remark}
\newtheorem*{remark}{Remark}

\theoremstyle{definition}

\begin{document}

\title{Avoiding symmetry roadblocks and minimizing the measurement overhead of adaptive variational quantum eigensolvers}

\author{V. O. Shkolnikov}
\affiliation{Department of Physics, Virginia Tech, Blacksburg, VA 24061, USA}
\affiliation{Virginia Tech Center for Quantum Information Science and Engineering, Blacksburg, VA 24061, USA}
\author{Nicholas J. Mayhall}
\affiliation{Department of Chemistry, Virginia Tech, Blacksburg, VA 24061, USA}
\affiliation{Virginia Tech Center for Quantum Information Science and Engineering, Blacksburg, VA 24061, USA}
\author{Sophia E. Economou}
\affiliation{Department of Physics, Virginia Tech, Blacksburg, VA 24061, USA}
\affiliation{Virginia Tech Center for Quantum Information Science and Engineering, Blacksburg, VA 24061, USA}
\author{Edwin Barnes}
\affiliation{Department of Physics, Virginia Tech, Blacksburg, VA 24061, USA}
\affiliation{Virginia Tech Center for Quantum Information Science and Engineering, Blacksburg, VA 24061, USA}

\maketitle

\begin{abstract}
Quantum simulation of strongly correlated systems is potentially the most feasible useful application of near-term quantum computers \cite{Preskill_quantum2018}. Minimizing quantum computational resources is crucial to achieving this goal. A promising class of algorithms for this purpose consists of variational quantum eigensolvers (VQEs) \cite{Peruzzo2014,McClean_njp2016,Romero_qst2018,Lee_jctcomp2019,CerezoNature2021,bharti2021noisy}. Among these, problem-tailored versions such as ADAPT-VQE \cite{Grimsley_naturecomm2019,Tang_PRXQuantum2021} that build variational ans\"atze step by step from a predefined operator pool perform particularly well in terms of circuit depths and variational parameter counts. However, this improved performance comes at the expense of an additional measurement overhead compared to standard VQEs. Here, we show that this overhead can be reduced to an amount that grows only linearly with the number $n$ of qubits, instead of quartically as in the original ADAPT-VQE. We do this by proving that operator pools of size $2n-2$ can represent any state in Hilbert space if chosen appropriately. We prove that this is the minimal size of such "complete" pools, discuss their algebraic properties, and present necessary and sufficient conditions for their completeness that allow us to find such pools efficiently. We further show that, if the simulated problem possesses symmetries, then complete pools can fail to yield convergent results, unless the pool is chosen to obey certain symmetry rules. We demonstrate the performance of such symmetry-adapted complete pools by using them in classical simulations of ADAPT-VQE for several strongly correlated molecules. Our findings are relevant for any VQE that uses an ansatz based on Pauli strings.
\end{abstract}

\section{Introduction}
Achieving an accurate description of strongly correlated systems is one of the most challenging scientific problems of our time \cite{Reiher_pnas2017,bassman_arxiv2021}. It manifests itself when the approximation in which electrons occupy orbitals independently from each other breaks down. In other words, strongly correlated materials cannot adequately be described by a single Slater determinant, no matter how we choose the orbital basis. In these cases, classical methods like density functional theory \cite{Hohenberg_physrev1964,Kohn_physrev1965} or mean-field approaches are not able to capture the behavior of the system. Although classical computational chemistry methods do exist (such as coupled cluster (CC)  methods) \cite{Taube_ijqc2006,Kutzelnigg_tca1991,Bartlett_chemphyslett1989} that go beyond the single Slater determinant, strong correlation quickly degrades their quantitative accuracy. For small enough systems, exact diagonalization (i.e.,  full configuration interaction (FCI)) can be used to obtain the ground and excited state energies to arbitrary precision. However, the dimension of the Hilbert space grows combinatorially, such that even for molecules of modest size, it becomes practically impossible to apply these methods. For example, the electronic wave function of a zero-spin molecule with twenty spatial orbitals (forty spin orbitals) at half-filling is represented by approximately 34 billion coefficients, which already takes about 275 Gb of memory and makes the calculations extremely challenging and time-consuming.

This "curse of dimensionality" can be avoided entirely if another, controllable quantum system is used to simulate the system of interest. For the example above, forty quantum bits would suffice to represent the many-body state we are interested in. If additionally the state of the simulator can be manipulated universally in any desired way, we can mimic the behavior of the original system under arbitrary conditions and thus gain access to its behavior and properties. Unfortunately, simulating the quantum system in this way requires a functional universal quantum computer, something that is not yet available. The biggest challenge here is that the quantum state of the simulator is prone to unwanted interactions with the environment that decohere the qubits and scramble its quantum state. As the current quantum machines have no error correction implemented \cite{Terhal_revmodphys2015}, it is impossible to perform long gate sequences on the state of the simulator; such sequences are essential for quantum algorithms that require universal control of the state \cite{Kitaev_2002, Nielsen_2000}. 

In this work, we focus on another paradigm of quantum simulation, one that trades universal control for shorter gate sequences. A promising class of simulation algorithms that operate in this paradigm are called variational quantum eigensolvers (VQEs) \cite{Peruzzo2014,McClean_njp2016,Romero_qst2018,Lee_jctcomp2019}. These algorithms are based on the standard variational principle of quantum mechanics, which states that the ground state energy is a global minimum of the expectation value of the Hamiltonian. The main steps of any VQE algorithm are as follows. First of all, one needs to map the Hilbert space of the system onto that of the simulator. There are a variety of ways to do the mapping, the most common being  Jordan-Wigner, Bravyi-Kitaev, and parity mappings \cite{McArdle_qcc2020}. In this work, we focus on the Jordan-Wigner mapping. VQEs operate by preparing the simulator system in some initial state $\ket{\psi}$ and then applying a sequence of parametrized unitary gates, $\hat{U}(\vec{\theta})$, referred to as an ansatz. The goal is then to find the minimum of $E(\vec{\theta})=\braket{\psi|\hat{U}(\vec{\theta})^\dagger\hat{H}\hat{U}(\vec{\theta})|\psi}$ with respect to $\vec{\theta}$, where $\hat{H}$ is the problem Hamiltonian expressed in terms of operators acting on the simulator. According to the variational principle, this minimum is the best approximation to the ground state allowed by the chosen ansatz. The advantage of VQEs is that the quantum hardware is only used to prepare the ansatz and measure the expectation value of $\hat{H}$, thus computing $E(\vec{\theta})$ for a given $\vec{\theta}$. The actual minimization procedure is performed on a classical computer that effectively uses the quantum simulator as a subroutine to compute $E(\vec{\theta})$. 

The performance of any VQE algorithm depends critically on the choice of the ansatz. In general, an ansatz is constructed from a product of individual quantum gates, each generated by some anti-Hermitian operator. The first question that arises is thus, "which sets of anti-Hermitian operators can serve in a VQE ansatz?" There are two general approaches taken in the community: (i) Use products of fermionic operators that resemble those appearing in the Hamiltonian of the system under study \cite{Peruzzo2014,Taube_ijqc2006,Bartlett_chemphyslett1989,EvangelistaJChemPhys2019}. This operator choice is typically what is used in classical chemistry simulations, but it has the disadvantage that when mapped to qubit operators and transpiled into the native gates of a quantum simulator, the circuit significantly deepens, and the implementation of the algorithm becomes problematic (although there has been recent progress in reducing these circuit depths \cite{YordanovPRA2020}). An alternative approach is the following: (ii) Use sets of anti-Hermitian operators that are native or at least easily implementable on the quantum simulator \cite{Kandala_nature2017,Cerezo_naturecomm2021}. Regardless of which approach is taken, there is no unique way to choose an ansatz, and there is currently no general answer as to what constitutes the best ansatz for a given problem. 

An important step toward addressing these issues was taken in Refs.~\cite{Grimsley_naturecomm2019,Tang_PRXQuantum2021}, which introduced an algorithm known as ADAPT-VQE that allows the system under study to determine its own problem-tailored ansatz by constructing it step by step from a predefined operator pool. It was shown that this leads to substantial reductions in circuit depths and variational parameter counts. Other types of iterative variational algorithms have also been proposed \cite{RyabinkinJCTC2018,rattew2020domainagnostic,chivilikhin2020mogvqe,gomes2021adaptive}. The original work on ADAPT-VQE \cite{Grimsley_naturecomm2019} utilized pools comprised of fermionic operators (fermionic-ADAPT-VQE), while follow-up work \cite{Tang_PRXQuantum2021} considered anti-Hermitian Pauli strings instead (qubit-ADAPT-VQE). Because Pauli strings can be transpiled into a small number of native gates, this choice leads to shorter ansatz circuits. Pools comprised of operators that lie somewhere in between these two cases have also been shown to yield good performance \cite{yordanov2020iterative}. The advantages afforded by ADAPT-VQE come at the expense of an increase in measurement overhead. This comes from the need to perform additional measurements (beyond those needed to obtain $E(\vec\theta)$) of the energy gradient to determine which operator to add in each iteration of the algorithm. The number of additional measurements is proportional to the size of the operator pool. Refs.~\cite{Grimsley_naturecomm2019,Tang_PRXQuantum2021,yordanov2020iterative} used pools of size $\mathcal{O}(n^4)$, where $n$ is the number of qubits, to perform molecular simulations. Because the number of terms in $\hat H$ for a molecule is also quartic, the total number of measurements per iteration would naively scale as $\mathcal{O}(n^8)$ instead of the $\mathcal{O}(n^4)$ scaling of a fixed-ansatz VQE applied to a molecular problem. However, by reformulating the energy gradient in terms of the 3-particle reduced density matrix, Ref.~\cite{LiuJChemPhys2021} showed that this measurement cost is actually $\mathcal{O}(n^6)$.\footnote{By accepting some amount of approximation to the energy gradient expressions, they also showed that an approximate 2-particle reduced density matrix reconstruction of the 3-particle reduced density matrix admits a measurement count that scales the same as the energy itself, $\mathcal{O}(n^4)$. Although this approximate gradient estimate created significant convergence problems, they have also developed a technique to mitigate these problems called ADAPT-Vx.}

An important consideration in choosing a pool is that it should be possible to construct from it an ansatz that represents the exact ground state to arbitrary accuracy. Pools that satisfy this criterion are called "complete" \cite{Tang_PRXQuantum2021}.
In Ref.~\cite{Tang_PRXQuantum2021}, it was shown that the $\mathcal{O}(n^4)$ pool can be reduced by over 90\% without sacrificing the performance of qubit-ADAPT-VQE, provided the completeness property is maintained. It was further proven (by constructing an explicit example) that there exist complete pools of size $2n-2$, and it was conjectured that this is the minimal size of such pools. This would in principle reduce the measurement overhead for molecules from $\mathcal{O}(n^8)$ to $\mathcal{O}(n^5)$. However, while such pools perform well for random Hamiltonians \cite{Tang_PRXQuantum2021}, they have not yet been tested on molecules or other problems of practical interest, so this overhead reduction has not been confirmed. Many additional questions remain as well: Is $2n-2$ really the minimal possible pool size or are there even smaller complete pools? Can we find all possible minimal complete pools, and how can we check if a pool is minimal and complete or not? Is it possible to reduce the necessary resources even further if one takes into account the details of the problem being solved, such as symmetries of the Hamiltonian? How does the choice of the minimal complete pool affect the convergence of the ansatz in the context of ADAPT-VQE?

In this paper, we address all of these questions. We prove that the minimum size of a complete pool comprised of Pauli strings is indeed $2n-2$. We also derive necessary and sufficient conditions for pool completeness that facilitate the process of constructing new examples of minimal complete pools. We further show that while generic minimal complete pools work well for random Hamiltonians, they generally do not perform well for molecular Hamiltonians. We show that this is due to the presence of symmetries in the latter. If symmetry-adapted minimal complete pools are used instead, then qubit-ADAPT-VQE successfully obtains the desired ground state of molecular problems while achieving the reduced $\mathcal{O}(n^5)$ measurement overhead. Although our main motivation is to improve the performance of ADAPT-VQE, it is important to stress that our results are relevant to any VQE that utilizes ans\"atze built from Pauli strings. Our findings allow one to determine when such ans\"atze are capable of exactly representing the desired state.

The paper is organized as follows. In the next section, we prove that minimal complete pools contain $2n-2$ operators, and we establish necessary and sufficient conditions for completeness. In Sec. 3, we then apply our results to random real Hamiltonians. In Sec. 4, we discuss minimal complete pools in the context of molecular simulations. We show that generic minimal complete pools lead to convergence difficulties. We then show that this problem is resolved by incorporating symmetries into the pool, and we analyze the performance of ADAPT-VQE with these symmetry-adapted minimal complete pools for several different molecules.

\section{Minimal complete pools}\label{mainText:sec:MCP}

In this section we discuss our main results related to minimal complete pools (MCPs). This section only gives an overview of the main theorems and does not present all proofs in detail. For a more rigorous discussion please see Appendix \ref{app:Minimal complete pools}. In this work, operator pools are defined as sets of anti-Hermitian Pauli strings $\{\hat{P}_i\}$, each of which is capable of generating a parameterized unitary $\exp(\alpha\hat{P}_i)$. We call an operator pool complete if for any two real states $\ket{\psi}$ and $\ket{\phi}$ there exists a product of these unitaries that can transform one to the other:
\begin{eqnarray}\label{mainText:Eq:completeness_definition}
\ket{\psi}=\prod_i\exp(\alpha_i\hat{P}_i)\ket{\phi}.
\end{eqnarray}
 Here, we focus on real states because we are primarily interested in simulating systems that possess time-reversal symmetry, such as molecules. This means that the $\hat P_i$ should contain odd numbers of Pauli $Y$ operators. In order to not work with imaginary matrix entries, we will use $iY$ instead of ordinary Pauli $Y$. Still, for clarity we will omit writing the factor $i$ in all expressions. The Pauli strings containing odd numbers of $Y$ operators will be referred to as odd Pauli strings ($O$ operators). In contrast, the Pauli strings with even numbers of $Y$ operators will be referred to as even ($E$ operators). We call a complete pool minimal if there is no complete pool of smaller size. In this section, we show how to identify all MCPs and discuss their main properties.

The basic features of an MCP are summarized by the following theorem:
\begin{theorem}\label{mainText:theorem:canonical_form_of_minimal_complete_pools}
	An MCP must contain $2n-2$ Pauli string generators $O_1$, $O_2$, ..., $O_{2n-2}$. If we use matrix multiplication to build the set of all possible products of these operators (a product group generated by an MCP), this set will coincide with the product group $G$ generated by $Z_1$, $Z_2$, ..., $Z_{n-2}$, $Y_1$, $Y_2$, ..., $Y_{n-2}$, $Y_{n-1}$, $Z_{n-1}Y_n$ up to a similarity transformation. The Lie algebra generated by an MCP is the subset of odd strings from this product group.
\end{theorem}
The similarity transformation referred to in the theorem above is defined as follows for any two odd Pauli strings $O_1$ and $O_2$ that anticommute:
\begin{equation}
	\exp\left(\frac{\pi}{4}O_1\right)O_2\exp\left(-\frac{\pi}{4}O_1\right)=\frac{1}{2}[O_1,O_2].
	\label{mainText:Eq: commutator_as_a_similatiry_transformation}
\end{equation}

Here, we briefly describe the main ideas behind the proof of Theorem \ref{mainText:theorem:canonical_form_of_minimal_complete_pools}. For a more detailed discussion, see Theorem \ref{theorem:canonical_form_of_minimal_complete_pools} and its proof in Appendix \ref{app:Minimal complete pools}. Completeness is related to the ability of the pool to generate a unitary that transforms the initial state of the system to the target ground state. The theory of Lie groups relates the unitaries a pool can generate to the size and structure of the Lie algebra generated by the pool through all possible commutators of its elements. We can understand this relationship in detail by taking advantage of the fact that our pool consists of Pauli strings. If any two Pauli strings do not commute, then their product coincides with their commutator up to a coefficient. This means that if we build a product group from a pool using matrix multiplication, then the Lie algebra will be a subset of this group. Therefore, if the product group is too small to host the operators necessary to generate the unitary we need, then the algebra will also be too small and so the pool cannot be complete. Thus in Appendix \ref{app:Minimal complete pools}, we first identify which operators we need in the group to be able to transform any real state to any other (as in Eq.~\ref{mainText:Eq:completeness_definition}), and then we use group theory to show that a minimal group that can host those operators must have the canonical form described in Theorem \ref{mainText:theorem:canonical_form_of_minimal_complete_pools}.
From now on, when we discuss an MCP, we will always assume it generates this canonical form of the product group $G$.

Theorem \ref{mainText:theorem:canonical_form_of_minimal_complete_pools} only considers the size and structure of the product group that an MCP must generate. This is a necessary condition for completeness, but at this point we have not yet shown that pools satisfying this condition exist. Such pools do in fact exist. In Appendix B of Ref.~\cite{Tang_PRXQuantum2021}, a pool of size $2n-2$ was constructed explicitly and shown to transform any real state into any other. One can check that up to renumbering of the qubits, this pool generates the group $G$ in Theorem \ref{mainText:theorem:canonical_form_of_minimal_complete_pools}.

If a pool generates an algebra that spans all odd strings from the group $G$ in Theorem \ref{mainText:theorem:canonical_form_of_minimal_complete_pools}, it will generate exactly the same algebra as the pool from Appendix B of \cite{Tang_PRXQuantum2021}, thus proving that this pool is also complete. This is a sufficient condition for completeness that is part of the following theorem, which is extremely useful when searching for minimal complete pools numerically:
\begin{theorem}[completeness criterion]
	Let a pool of $2n-2$ Pauli string generators $O_1$, $O_2$, ..., $O_{2n-2}$ generate the product group $G$ defined in Theorem \ref{mainText:theorem:canonical_form_of_minimal_complete_pools}. The following statements are equivalent:
	\begin{itemize}
		\item (a) The pool $O_1$, $O_2$, ..., $O_{2n-2}$ is complete.
		\item (b) The pool $O_1$, $O_2$, ..., $O_{2n-2}$ cannot be split into two mutually commuting sets.
		\item (c) The algebra generated by $O_1$, $O_2$, ..., $O_{2n-2}$ spans all odd strings from the group $G$.\\
	\end{itemize}
\label{mainText:theorem:completeness_criterion}
\end{theorem}

Statement (b) is a necessary condition of completeness and thus follows from (a), as shown in Appendix \ref{app:Minimal complete pools}. What this statement means is that we cannot split the pool into two sets of operators, such that each operator from the first set commutes with each operator from the second set. As already discussed above in the context of Theorem \ref{mainText:theorem:canonical_form_of_minimal_complete_pools}, statement (a) follows from (c). Proving that (c) follows from (b) turns out to be very challenging. For now we do not have an analytical proof, but all our numerical calculations confirm this is true, and we use this statement in practice.

\begin{remark}
Theorem \ref{mainText:theorem:completeness_criterion} consists of completeness criteria that are in part proven analytically and in part strongly supported numerically. It coincides with Theorem \ref{theorem:completeness_criterion} of Appendix \ref{app:Minimal complete pools}. The condition (c) can safely be used to search for complete pools, as its applicability has been proven analytically. However, computing the Lie algebra for a given pool is very resource demanding. If one needs to check many pools for completeness and select one based on some other criterion, this approach might take too long. This is why condition (b) is extremely useful, as its computational complexity scales polynomially with the size of the pool, and thus it is a lot easier and faster to use.
\end{remark}
\begin{remark}
The above theorems and observations lead to a practical recipe to search for complete pools. The problem statement is the following: given a pool of $2n-2$ Pauli strings, check whether it is complete or not. 
\begin{itemize}
	\item Step 1. Generate the product group and check if it coincides with that of Theorem \ref{mainText:theorem:canonical_form_of_minimal_complete_pools} up to a similarity transformation. This can be done using Theorem \ref{theorem:necassary_condition_of_completeness} from Appendix \ref{app:Minimal complete pools}, i.e., by checking that the group contains $2^n-1$ odd Pauli strings that perform all possible flippings of the qubits. \\
	\item Step 2. If the group is correct, the next step is to check if the pool obeys the inseparability criterion (condition (b) of Theorem \ref{mainText:theorem:completeness_criterion}).\\ 
	\item Step 3. One could already stop here. However, if one wants to rely on a fully analytical completeness proof, one needs to compute the algebra generated by the pool and check that it spans all odd strings from the group $G$. Equivalently, one needs to check that the algebra size is $\frac{2^{n-1}(2^{n-1}+1)}{2}$, in agreement with Lemma \ref{lemma:number_of_odd_strings_in_complete_group} of Appendix \ref{app:Minimal complete pools}. \\
\end{itemize}
\label{remark: steps_to_check_pools_for_completeness}
\end{remark}

\section{ADAPT-VQE with minimal complete pools for random Hamiltonians}
Now we will apply the results of the previous section to dense random Hamiltonians. Thus, we will assume the simulator can perform the gates $\exp(\alpha_i \hat{P}_i)$, parametrized by a real number $\alpha_i$, where $\hat{P}_i$ is taken from an operator pool $\{\hat{P}_1, \hat{P}_2, ..., \hat{P}_k\}$. In order to perform a simulation, we need to specify how we will construct the actual ansatz. We follow the protocol of ADAPT-VQE, which, in each iteration of the algorithm, selects the pool operator that has the largest energy gradient and adds it to the ansatz \cite{Grimsley_naturecomm2019,Tang_PRXQuantum2021}. We first initialize the system in a classical reference state $\ket{\psi^{(0)}}$. For molecules, this is typically a Hartree-Fock state, while for random Hamiltonians, we just initialize all qubits to the state $\ket{0}$. We then select an operator $\hat{P}_i$ from the pool by measuring the commutator $\braket{\psi^{(0)}|\big[\hat{H},\hat{P}_i\big]|\psi^{(0)}}$ for each pool operator and choosing the largest one. The ansatz then becomes $\ket{\psi(\alpha_i)}=\exp(\alpha_i \hat{P}_i)\ket{\psi^{(0)}}$.  The commutators are proportional to the derivative of the energy $\braket{\psi(\alpha_i)|\hat{H}|\psi(\alpha_i)}$ with respect to $\alpha_i$, justifying this selection criterion. In the next step we minimize $\braket{\psi(\alpha_i)|\hat{H}|\psi(\alpha_i)}$ with respect to $\alpha_i$ using a standard VQE procedure, involving both quantum and classical hardware. The resulting state $\ket{\psi^{(1)}}=\ket{\psi(\alpha_i^*)}=\exp(\alpha_i^* \hat P_i)\ket{\psi^{(0)}}$ then replaces $\ket{\psi^{(0)}}$ in the next iteration. We use $\ket{\psi^{(1)}}$ in the gradient criterion to choose the next operator from the pool. Our ansatz will then take the form $\ket{\psi(\vec{\alpha})}=\exp(\alpha_j \hat P_j)\exp(\alpha_i \hat P_i)\ket{\psi^{(0)}}$. Note that after choosing the next operator from the pool, we unfreeze the coefficient $\alpha_i$ and optimize with respect to the vector $\vec{\alpha}=(\alpha_j,\alpha_i)$. In real simulations on quantum hardware, the loop terminates when the energy gradient is zero or when the energy does not decrease significantly after several steps of the algorithm. Here, however, we focus on classically tractable problems to test the performance of various pools; in this case, we can compare the resulting energy to the true (FCI) ground state energy obtained from exact diagonalization and terminate the loop when the difference between them is below a predefined threshold. This can only work of course for small molecules or Hamiltonians that can be diagonalized numerically.

To illustrate the theory from the previous section, we randomly generate real Hamiltonians for $n=6$ and $n=8$ qubits. We use MCPs selected randomly. To create these, we take random sets of $2n-2$ operators and check if they generate the correct product group and obey the inseparability criterion. To be on the safe side, we also generate the algebra and make sure it spans all odd Pauli strings from the product group (following the steps in the second remark below Theorem \ref{mainText:theorem:completeness_criterion}). In the case of 6 qubits, one such pool has the following 10 operators:
\begin{equation}
	\begin{aligned}
		&\text{XZIIXY}, \text{ZXYZII}, \text{YZYYII}, \text{YYIIXY}, \text{IZXXZY},\\
		&\text{XZIXZY}, \text{ZYIYYI}, \text{XIYYYI}, \text{YIYZYI}, \text{XYZYYI},\\
	\end{aligned}
	\label{Pool:random_six_qubit_Hamiltonian}
\end{equation}
while in the case of 8 qubits, the pool has 14 operators, for example:
\begin{equation}
	\begin{aligned}
		&\text{ZYIZIYZY}, \text{ZXXZYYYI}, \text{YZIIIXII}, \text{YZXYIIXY},\\
		&\text{IIXXIIYI}, \text{IYYYZZII},\text{IYXZIYZY}, \text{ZXZIIXYI},\\
		&\text{YYZZZIYI}, \text{YIXYZZXY},\text{IIXXXIYI}, \text{IYXXIYXY},\\
		& \text{ZYIXIXII}, \text{XYXIZZII}.
	\end{aligned}
	\label{Pool:random_eight_qubit_Hamiltonian}
\end{equation}
\begin{figure}[h]
	\centering
	\includegraphics[width=0.4\textwidth]{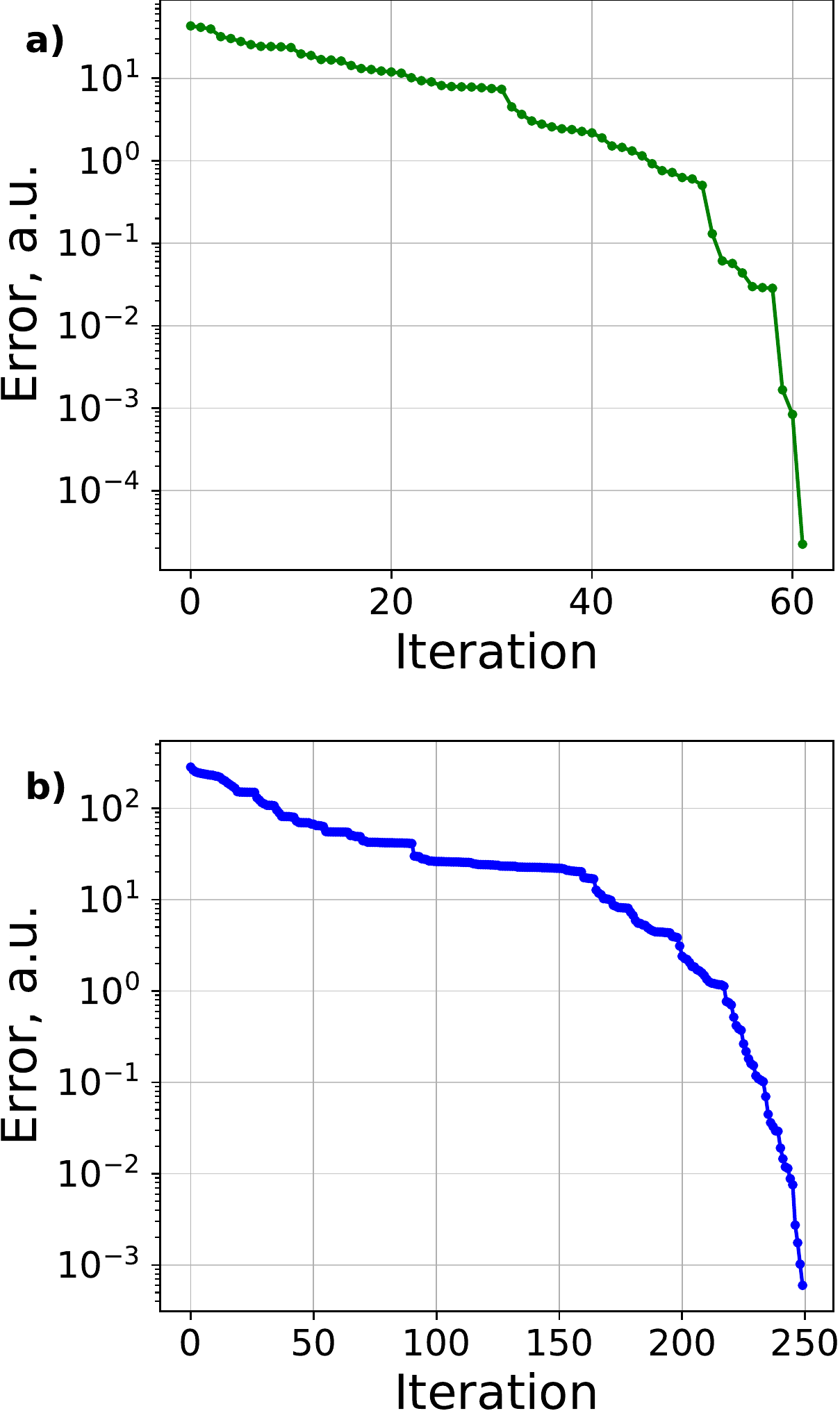}
	\caption{Energy error versus ADAPT-VQE iteration number for random Hamiltonians for (a) $n=6$ and (b) $n=8$ qubits, and randomly chosen pools (see Eqs.~\eqref{Pool:random_six_qubit_Hamiltonian} and \eqref{Pool:random_eight_qubit_Hamiltonian}). The error is defined as the difference between the energy at a particular iteration and the FCI energy.}
	\label{Fig:random_hamiltonians_error_vs_iteration}
\end{figure}
	Figure~\ref{Fig:random_hamiltonians_error_vs_iteration} shows the energy error (relative to the FCI energy) versus iteration number for two examples of randomly generated Hamiltonians. The curve follows a similar pattern for the two cases considered, and ADAPT-VQE manages to decrease the error by six orders of magnitude in each case. 
	It is evident from the figure that the number of parameters needed to reach convergence quickly increases from $\sim60$ to $\sim250$ as the number of qubits increases from $n=6$ to 8. This is a consequence of the randomness of the Hamiltonians. The lack of structure in the ground states of such Hamiltonians means that they do not admit an efficient representation, and the number of parameters needed likely scales with the Hilbert space dimension. We will see below that for physical systems such as molecules, where the Hamiltonian and ground state possess structure, the number of parameters needed to reach convergence is much smaller than the Hilbert space dimension. There, we will see that a judicious choice of MCP can significantly speed up the algorithm and further reduce parameter counts. In fact, choosing a random MCP to search for the ground state energy only works for random Hamiltonians. Real molecular Hamiltonians have a lot of structure due to symmetries. In this case, a random MCP will fail to achieve convergence, and one has to take into account symmetries of the problem Hamiltonian to construct a suitable MCP that is capable of producing the target ground state. We discuss these issues in the next section.

\section{Symmetry-preserving minimal complete pools for molecular simulations}
In this section, we apply MCPs to search for the ground state of molecular Hamiltonians using qubit-ADAPT-VQE \cite{Tang_PRXQuantum2021}. We begin with a discussion of the H$_4$ molecule, before developing a more general theory of molecular simulations using MCPs. H$_4$ is a linear molecule that consists of four protons and four electrons. For simplicity, we only retain four 1s molecular orbitals (one for each hydrogen atom) to span the state space. This leaves us with eight spin orbitals that, under the Jordan-Wigner mapping, translate to eight qubits. For the initial state, we choose the classical Hartree-Fock state. When we run qubit-ADAPT-VQE with a random MCP, we encounter a problem, namely the algorithm does not start. In particular, all the gradients in the first step are exactly zero, and no operator is identified for inclusion into the ansatz. We note that this problem is specific to the gradient criterion of ADAPT-VQE, as in principle any MCP can, according to our theory, take us to the ground state. The gradient criterion of ADAPT-VQE in this case is just unable to provide the corresponding path. This can be explained simply by the very sparse structure of the Hamiltonian. Indeed, the number of terms in molecular Hamiltonians is quartic in the number of qubits ($n^4$) due to the fact that such Hamiltonians only contain one and two-particle terms. This sparse structure makes it likely that terms in the Hamiltonian will commute with the operators from the pool, thus making the energy gradient zero. Still, the fact that \textit{all}  $n^4$ terms commute with all the operators from the pool must have a different explanation. If we just put $n^4$ random Pauli strings in a Hamiltonian, the probability that all of them would commute with the pool would be extremely low. That means there is a pattern among the Pauli strings that comprise the Hamiltonian, and it turns out this pattern is completely defined by the Hamiltonian symmetries. 

First of all, every molecular Hamiltonian conserves the spin and the number of particles in the system. This means the Hamiltonian has vanishing matrix elements between states of different spin or different particle number. The H$_4$ molecule has an additional symmetry on top of these, namely inversion symmetry. This implies all eigenstates of the Hamiltonian are either symmetric or antisymmetric with respect to inversion, and the Hamiltonian will have vanishing matrix elements between any two states of different parity. This can be more strictly formulated using the theory of group representations; we employ this approach later, but for now we only discuss H$_4$ in these simple terms. We start the simulation at the Hartree-Fock state, which has four electrons, zero spin, and a well defined parity. In terms of qubits in the simulator, this state can be expressed as $\ket{\psi_{HF}}=\ket{11110000}$. According to the discussion above, if a Pauli string $\hat{P}$ changes the number of particles, spin or parity of this state, its  commutator  with the Hamiltonian is exactly zero:
\begin{equation}
	\begin{aligned}
 &\braket{\psi_{HF}|[\hat{P},\hat{H}]|\psi_{HF}}=\\
 &=\braket{(\hat{P}\psi_{HF})|\hat{H}|\psi_{HF}}-\braket{\psi_{HF}|\hat{H}|(\hat{P}\psi_{HF})}=0.\\
	\end{aligned}
\label{Eq:Zero_commutator_1}
\end{equation}
This suggests that when selecting a pool, at least some operators in it must conserve the number of particles, spin, and parity of the Hartree-Fock state in order for the algorithm to start. These are not all the limitations we must consider when constructing the pool. Another limitation comes from the fact that a Hartree-Fock state is the classical state that lies closest in energy to the ground state. If we consider an operator $\hat{P}$ that conserves all the quantum numbers and at the same time is a single-particle excitation (contains exactly two $X$ or $Y$ operators), the commutator $\braket{\psi_{HF}|[\hat{P},\hat{H}]|\psi_{HF}}$ would again be zero. Otherwise one would be able to construct a different classical state with energy lower than that of the Hartree-Fock state. Indeed, let us assume $\braket{\psi_{HF}|[\hat{P},\hat{H}]|\psi_{HF}}$ is nonzero for $\hat{P}=IIIYIIIX$. That means the matrix element
\begin{equation}
	\begin{aligned}
		&\braket{\psi_{HF}|\hat{H}\hat{P}|\psi_{HF}}=\braket{11110000|\hat{H}|11100001}
	\end{aligned}
	\label{Eq:Zero_commutator_2}
\end{equation}
is nonzero as well. In that case there must exist a superposition state, 
\begin{equation}
	a\ket{11110000}+b\ket{11100001},
\end{equation}
that has lower energy than the Hartree-Fock state $\ket{11110000}$. This would mean that one electron resides in a superposition of orbitals 2 and 4, and if one constructs a new basis containing this superposition state, a product state with energy lower than that of the Hartree-Fock state would be obtained (an extension of the Brillouin condition to Pauli operators). 

Finally, the commutator $\braket{\psi_{HF}|[\hat{P},\hat{H}]|\psi_{HF}}$ will vanish if we consider a $\hat{P}$ that creates more than double excitations on top of the Hartree-Fock state. This is a consequence of the molecular Hamiltonian containing only single and double excitations, so it cannot connect for example the states $\ket{11110000}$ and $\ket{00001111}$. 

These rules fully define the operators that will have nonzero gradient when qubit-ADAPT-VQE starts. More generally, these operators can quickly lower the energy at the beginning of the algorithm, so it makes sense to include them in the pool even when some other VQE procedure is used that does not directly depend on the energy gradients.  We refer to the operators obeying the rules above as "starters". They allow the state to evolve away from the Hartree-Fock state in the initial steps of the algorithm, thus providing Pauli contributions to the "first-order interacting space", which is typically expressed in terms of second-quantized fermionic operators. In our simulations, we choose at least half of the operator pool to consist of starters. 

Even if we choose the pool to contain some starters and then add in other operators to satisfy the completeness criterion, we run into a different problem. It turns out all the gradients will often zero out again, before the ADAPT-VQE algorithm reaches the ground state. Figure~\ref{Fig:MCP_no_symmetry_but_starters} demonstrates this behavior in the case of H$_4$. An MCP of 14 operators in total is used, 7 of which are starters. The energy error does not improve beyond $2\times10^{-3}$ Ha as the gradient goes to zero. Although we do not show it here, similar behavior also arises for LiH and BeH${}_2$.
\begin{figure}[h]
	\centering
	\includegraphics[width=0.45\textwidth]{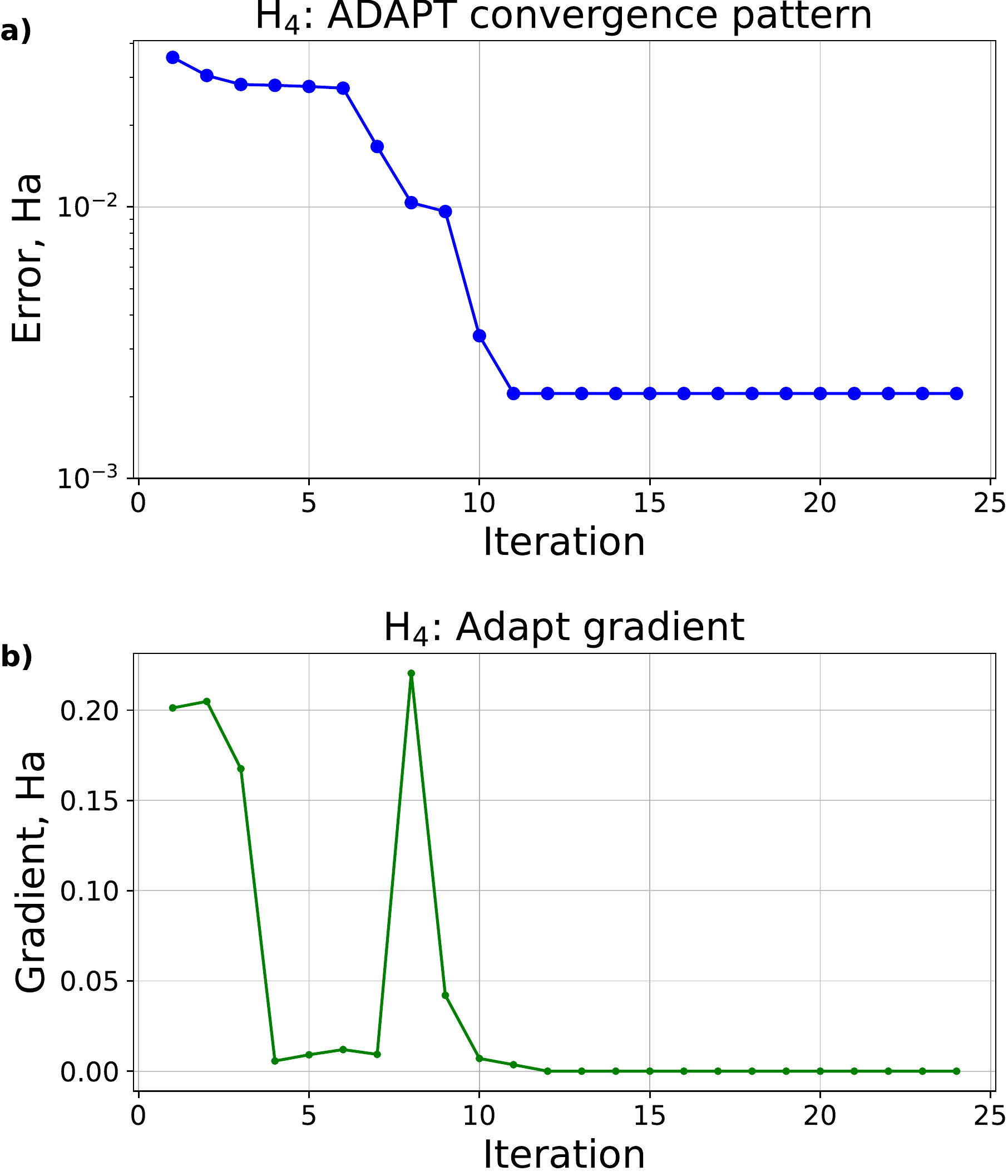}
	\caption{The absolute error (a) and maximal gradient (b) at each step of the qubit-ADAPT-VQE algorithm for the H$_4$ molecule mapped onto 8 qubits. The pool is chosen to be a 14-operator MCP with 7 starters, in which symmetry is not taken into account. The gradient goes to zero when the absolute error is only $2\times10^{-3}$ Ha. }
	\label{Fig:MCP_no_symmetry_but_starters}
\end{figure}
This can again be understood as a consequence of symmetry, as we now explain. 

Let us again focus for the moment on the H$_4$ molecule and introduce several new concepts. We will refer to a state having an even/odd number of electrons as a state of positive/negative particle-number parity. This new parity comes on top of the usual parity related to spatial inversion. Now one can show that every Pauli string that changes at least one of these parities will always have zero gradient and will never be picked. This can be proved by induction on the iterations of the ADAPT-VQE algorithm. Let us assume the statement is true for the first $N$ steps of the algorithm. That is, no operator that would change one of the parities was picked, so we know that after $N$ steps of the algorithm, the parities will still coincide with those of the initial Hartree-Fock state. Without loss of generality, we can assume these parities to be positive so that the total parity state is $\ket{+,+}$. The Hamiltonian will only connect this state to other states of the same parity. Note that when acting on a state for which one of the parities is a good quantum number, a Pauli string will either change this parity or leave it intact. Suppose we compute an energy gradient for a Pauli string that changes parities, so for example $\ket{+,+}\rightarrow\ket{+,-}$. Then
\begin{equation}
	\begin{aligned}
		&\braket{+,+|[\hat{H},\hat{P}]|+,+}=\\
		&2\text{Im}\braket{+,+|\hat{H}|+,-}=0.
	\end{aligned}
	\label{Eq:Zero_commutator_3}
\end{equation}
This means such a Pauli string will not be picked, and our statement is also true for the step $N+1$. The base of induction is the fact that at step zero of the algorithm, no operator was picked at all, so the statement is true. 

Now we can explain why ADAPT-VQE stops before it converges, even though we are using a complete pool (more specifically, an MCP). Our definition of completeness requires that we can reach any point in the Hilbert space from any other point. That means we must have operators that will change parities. But they are never picked by the ADAPT-VQE gradient criterion and thus are never added to the ansatz, so effectively our pool is smaller and thus it might be incomplete in practice. We again note that this result does not mean that there is no path to the ground state, it just means the ADAPT-VQE gradient criterion is not able to find the path in this case; another ADAPT-VQE criterion or other  VQE algorithms could in principle succeed with the same pool. However, resorting to a different operator-selection criterion or algorithm is not necessary, as we can actually still use the ADAPT-VQE gradient criterion to solve our problem if we notice that the MCP requirement that we can go from any point in the Hilbert space to any other point is actually too strict. Indeed, going from the Hartree-Fock state to the ground state can be written as
\begin{equation}
	\begin{aligned}
		&\ket{\psi_{gs}}=\exp(\hat{M})\ket{\psi_{HF}},
	\end{aligned}
	\label{Eq:path_to_the_ground_state}
\end{equation}
where $\hat{M}$ is a real antisymmetric matrix. This matrix, according to Lie group theory, is a superposition of operators from the complete algebra generated by the MCP. On the other hand, the operators present in this superposition obey all the symmetry requirements, namely they do not change the number of particles, the spin, or the parity of the Hartree-Fock state (this is a result of the fact that these quantum numbers are exactly the same for the Hartree-Fock state as they are for the ground state of the system). But then they will also keep both parities intact. Consider the subgroup of the product group and the subalgebra of the complete algebra that consist of operators conserving the parities. It follows that $\hat M$ is a superposition of the operators from this subalgebra (and thus from the subgroup). It is now clear that the most efficient pools should consist of odd Pauli strings that belong to this subalgebra and thus preserve parity. Due to the binary nature of the parities, the size of the subgroup is 4 times smaller than that of the group generated by a generic MCP. That means we can expect the parity-preserving pool of operators to contain 2 fewer operators compared to a general MCP. The importance of restricting ansatz generators according to symmetry considerations was also emphasized in Ref.~\cite{RyabinkinJCTC2020} in the context of the Qubit Coupled Cluster algorithm \cite{RyabinkinJCTC2018}. 

Is this parity-preserving pool the smallest that will work or can we reduce the size of the pool even further? If we could do this, then it would mean the matrix $\hat{M}$ is a superposition of Pauli strings that belong to an even smaller subalgebra of the original algebra. This turns out to be exactly the case, as we only used parities in our above considerations and have not yet taken into account spin symmetry ($\hat{S}_z$). In order to do so, let us first notice that Pauli strings come in one of two symmetry types. The first type are those Pauli strings that are not able to conserve spin and particle number at the same time for any classical state. An example of such a Pauli string in the four-qubit case is $\text{XYII}$ if the spin orbitals are ordered with alternating spin, e.g., $\alpha\beta\alpha\beta$, where $\alpha$ and $\beta$ represent the $1/2$ and $-1/2$ spin projection, respectively, for the electron occupying the corresponding orbital. Indeed, this Pauli string will only conserve the number of particles if the classical state belongs to one of the two cases:
\begin{equation}
    \ket{01...}\qquad
    \text{or}\qquad
    \ket{10...}.
\end{equation}
In either of these cases, the spin will change by 1, so particle number and spin cannot be conserved at the same time by this Pauli string for any classical state. The second type of Pauli strings are those that are able to conserve spin and particle number at the same time. These are exactly those Pauli strings for which the number of $X$ or $Y$ operators acting on the $\alpha-$orbitals is a multiple of two (same is true for $\beta$ orbitals). In other words, $X$ and $Y$ operators come in pairs for a given spin projection. In our four-qubit case, an example of such a Pauli string is $\text{YIXI}$, as it contains two $X$ or $Y$ Pauli operators acting on the $\alpha$ orbitals and no operators acting on the $\beta$ orbitals. This observation makes it clear that Pauli strings of the second type form a group under matrix multiplication. The matrix $\hat{M}$ conserves the spin and particle number of the Hartree-Fock state, so it must be a superposition of Pauli strings belonging to the second type. That means the parity-preserving subgroup we identified above can be further restricted by these conditions. This reduces the size of the symmetry-preserving subgroup by another factor of 2, and hence reduces the pool size by one more operator.  

We can now precisely formulate how to choose the pool in order to ensure convergence to the target ground state. It must generate the symmetry-preserving product subgroup and subalgebra that we described above, as well as contain enough starters to allow ADAPT-VQE to start. In other words, the pool must obey the following conditions:
\begin{enumerate}\label{enum:symmetry_constraints_on the pool}
    \item The number of electrons with a given spin changes by a multiple of 2. That is, there is an even number of $X$, $Y$ operators acting on $\alpha$ orbitals, and there is an even number of $X$, $Y$ operators acting on $\beta$ orbitals.
    \item Each operator in the pool must conserve spatial parity. This condition will be generalized below for more complicated molecules.
    \item The pool must contain enough starters for ADAPT-VQE to start.
    \item The pool generates the biggest subgroup and subalgebra of those generated by a general non-symmetry-preserving MCP, that contain Pauli strings obeying conditions 1-2.
\end{enumerate}
We refer to pools that satisfy these criteria as symmetry-preserving MCPs.

 We now construct a pool obeying the conditions above for the H$_4$ molecule. Due to the binary nature of parity, the size of the subgroup is 8 times smaller than that of the group generated by the general, non-symmetry-preserving MCP. The pool we construct thus contains 3 fewer operators than the general MCP, which leaves us with 11 operators in the pool (see Eq.~\eqref{Pool:random_eight_qubit_Hamiltonian}). We do not impose any further restrictions on the choice of the pool---the remaining choices are conducted randomly. For our simulation of H$_4$ dissociation curve (Fig.~\ref{Fig:H4_dissociation_curve}), we use the following symmetry-preserving MCP:
\begin{equation}
	\begin{aligned}
		&\text{YIXIYIYI}, \text{ZYXIYIZY}, \text{YIZYXIZY}, \text{ZZYXYYII},\\
		&\text{XXIZIIXY}, \text{YIZYZXYI},\text{XIYZYZYI}, \text{XZIIYZII},\\
		&\text{ZXXZZXYI}, \text{XXIIIIXY}, \text{IYYZXIZY}.
	\end{aligned}
	\label{Pool:H4}
\end{equation}

Let us analyze this pool in more detail. In order to do that we have to first understand how our basis orbitals look in terms of symmetry. If we choose the Hartree-Fock state of the H$_4$ molecule to have the first four qubits in state $\ket{1}$, $\ket{\psi_{HF}}=\ket{11110000}$, the ordering of the orbitals in terms of parity and spin can be chosen to be
\begin{equation}
    \ket{\underbrace{\alpha\beta}_{\text{+}}\underbrace{\alpha\beta}_{\text{-}}\underbrace{\alpha\beta}_{\text{+}}\underbrace{\alpha\beta}_{\text{-}}}.
    \label{eq:H4_molecular_orbitals}
\end{equation}
The underbrace shows whether the corresponding orbital has a positive or negative spatial parity. Looking at the orbital structure in Eq.~\eqref{eq:H4_molecular_orbitals}, one infers that all the operators from the pool in Eq.~\eqref{Pool:H4} except the 8th one, $\text{XZIIYZII}$, are starters. Indeed, they all conserve the number of particles, spin and parity of the Hartree-Fock state, as well as contain exactly four $X$ and $Y$ operators. The operator $\text{XZIIYZII}$ is a single-excitation operator, so it is not a starter even though it conserves all three symmetries (Brillouin condition). In Fig.~\ref{Fig:H4_dissociation_curve} we show the dissociation curve for the H$_4$ molecule, computed using the pool from Eq.~\eqref{Pool:H4}. At all bond lengths, ADAPT-VQE converges to FCI with an error less than $10^{-8}$ Ha in about 60 steps. The symmetry-preserving pool clearly outperforms the symmetry-violating MCP used in Fig.~\ref{Fig:MCP_no_symmetry_but_starters} even though the symmetry-preserving one contains fewer operators. This clearly highlights the importance of incorporating symmetry considerations into the pool. 

To further examine the role of starters, we can construct and compare several MCPs containing different numbers of them following the approach used to construct the pool in Eq.~\eqref{Pool:H4}. Using the three starters
\begin{equation}
	\begin{aligned}
		&\text{ZYXZZYYI}, \text{YZIYZXYI}, \text{XYZZYYII},
	\end{aligned}
	\label{Pool:H4_3_starters}
\end{equation}
gives rise to the following MCP:
\begin{equation}
	\begin{aligned}
		&\text{ZYXZZYYI}, \text{XIIIYZII}, \text{YZIYZXYI}, \text{ZIXZXXZY},\\
		&\text{XYZZYYII}, \text{ZXXXYZII},\text{IXYYZXXY}, \text{IZZXIZZY},\\
		&\text{YYXIXXYI}, \text{XYXXZIII}, \text{ZYXXIYXY}.
	\end{aligned}
	\label{Pool:H4_3_starters_pool}
\end{equation}
On the other hand, using the six starters
\begin{equation}
	\begin{aligned}
		&\text{YIIXYZZY}, \text{IXXZXIZY}, \text{ZZXYYYII},\\
		&\text{YZZYXZZY}, \text{IXXZIXYI},\text{YYZIIIXY},
	\end{aligned}
	\label{Pool:H4_6_starters}
\end{equation}
leads to a different MCP:
\begin{equation}
	\begin{aligned}
		&\text{YIIXYZZY}, \text{IXXZXIZY}, \text{ZZXYYYII}, \text{YZZYXZZY},\\
		&\text{IXXZIXYI}, \text{YYZIIIXY},\text{ZIXIIZYI}, \text{IYYYIYXY},\\
		&\text{YIZIXZII}, \text{ZZXIZZYI}, \text{ZXZXXIYI}.
	\end{aligned}
	\label{Pool:H4_6_starters_pool}
\end{equation}
The following nine starters
\begin{equation}
	\begin{aligned}
		&\text{ZYYZXIZY}, \text{IXXIXZZY}, \text{YIZXIYYI},\\
		&\text{IXYIYIZY}, \text{IZYXYYII},\text{XXIZXYII},\\
		&\text{YYIZXYII},
		\text{IXZYIYZY}, \text{XXZIXYII},
	\end{aligned}
	\label{Pool:H4_9_starters}
\end{equation}
yield this MCP:
\begin{equation}
	\begin{aligned}
		&\text{ZYYZXIZY}, \text{IXXIXZZY}, \text{YIZXIYYI}, \text{IXYIYIZY},\\
		&\text{IZYXYYII}, \text{XXIZXYII},\text{YYIZXYII}, \text{IXZYIYZY},\\
		&\text{XXZIXYII}, \text{XZYIZYZY}, \text{IIIXXXYI}.
	\end{aligned}
	\label{Pool:H4_9_starters_pool}
\end{equation}

We use the pools in Eqs.~\eqref{Pool:H4}, \eqref{Pool:H4_3_starters_pool}, \eqref{Pool:H4_3_starters_pool}, \eqref{Pool:H4_9_starters_pool} to calculate the ground state energy of the H$_4$ molecule at the minimum of the dissociation curve, with the results shown in Fig.~\ref{Fig:H4_dissociation_curve}(b). The pool without starters does not allow the system to converge, and so that result is not shown here. The pool with three starters converges more slowly and has a longer plateau compared to using the pools with 6 or more starters. This illustrates the role and importance of starters when choosing the pool. On the other hand, all MCPs approach full convergence around the same number of iterations ($\sim30$ in this case), and so the differences between the MCPs may only matter if one stops the calculation before full convergence is reached. This could be the case, for instance, if achieving chemical accuracy is sufficient. In such cases, using an MCP with more starters is advantageous.
\begin{figure}[h]
	\centering
	\includegraphics[width=0.45\textwidth]{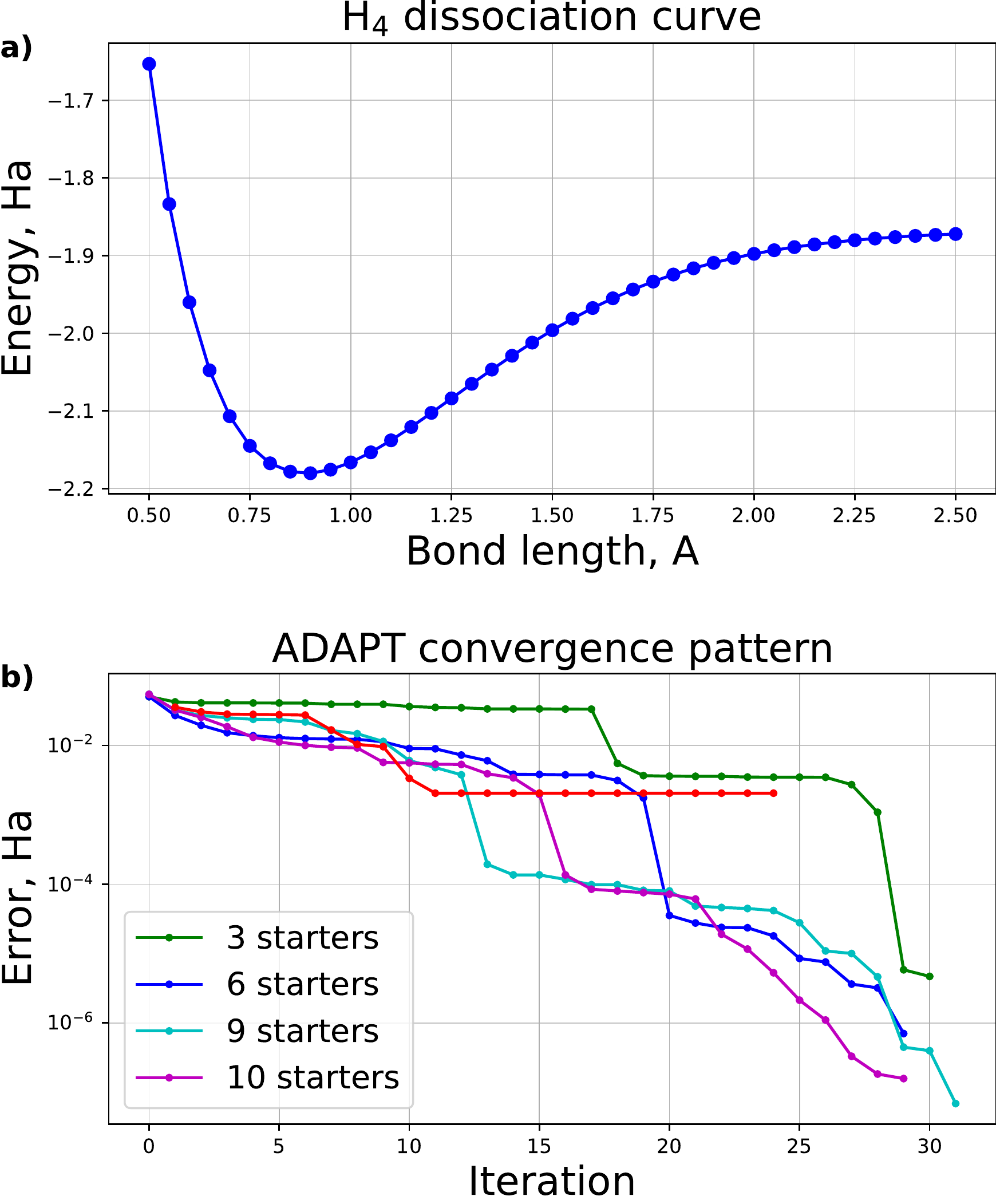}
	\caption{(a) Dissociation curve for the H$_4$ molecule obtained using qubit-ADAPT-VQE on 8 qubits with a symmetry-preserving 11-operator MCP (Eq.~\eqref{Pool:H4}). The dots are the results of individual VQE simulations, while the line is a guide to the eye. The absolute error at each bond length is less than $10^{-8}$ Ha. (b) Absolute error vs. iteration number of the ADAPT-VQE simulation at the minimum of the dissociation curve computed for the different symmetry preserving MCPs listed in Eqs.~\eqref{Pool:H4}, \eqref{Pool:H4_3_starters_pool}, \eqref{Pool:H4_6_starters_pool}, \eqref{Pool:H4_9_starters_pool} (the pools differ by the number of starters included). For comparison, we also include the result from Fig.~\ref{Fig:MCP_no_symmetry_but_starters} (red line). One can clearly see that when a symmetry-adapted pool is used, the error reduces substantially, and an insufficient number of starters makes the convergence slower.}
	\label{Fig:H4_dissociation_curve}
\end{figure}

We now discuss how to account for symmetries more generally. For more complicated molecules, the simple picture above cannot be applied directly, and we need to invoke group theory to construct the correct pool. More precisely, our arguments regarding the number of particles and spin will always be valid for any molecule (condition 1 above), but the spatial symmetry will generally be more complicated than the simple inversion symmetry we have for H$_4$, and group representation theory is required. The following discussion relies on the group theoretical notation used for example in Ref.~\cite{Tinkham_1964}. In order not to overcomplicate the discussion, we will concentrate on the most frequently encountered case, when the ground state is a non-degenerate fully symmetric state. In other words, this state belongs to the $A_1$ irreducible representation of the molecular symmetry group. We assume the Hartree-Fock state to belong to the same $A_1$ irreducible representation. Let us first describe how the starters should look. In the case of the H$_4$ molecule, we said that the starters must conserve the parity of the Hartree-Fock state. Generalizing this statement requires the application of the Wigner-Eckart theorem. Every Pauli string operator can be thought of as a sum of symmetric components, 
\begin{equation}
	\begin{aligned}
		&\hat{P}=\bigoplus_i\hat{S}_i,\\
	\end{aligned}
	\label{Eq:Pauli_string_symmetry_decomposition}
\end{equation}
where each $\hat{S}_i$ transforms according to the $i_{th}$ irreducible representation of the symmetry group ($R_i$). In order for the Pauli string $\hat{P}$ to be a starter, the commutator $\braket{\psi_{HF}|[\hat{P},\hat{H}]|\psi_{HF}}$ must be nonzero, which means that the Hamiltonian must have a nonzero matrix element between the states $\ket{\psi_{HF}}$ and $\hat{P}\ket{\psi_{HF}}$. According to the Wigner-Eckart theorem, this will hold true if $A_1\in R_i\otimes A_1$ for some $i$. Because we assume $A_1$ to be the fully symmetric irreducible representation, it follows that $R_i\otimes A_1=R_i$ and
\begin{equation}
	\begin{aligned}
		&A_1\in R_i\otimes A_1 \iff A_1\in R_i \iff R_i=A_1,
	\end{aligned}
	\label{Eq:condition_for_the_starters}
\end{equation}
so the fully symmetric component must be present in the symmetry decomposition of the Pauli string $\hat P$ (Eq.~\eqref{Eq:Pauli_string_symmetry_decomposition}).

The next question we need to answer is whether we can restrict the complete group and algebra based on symmetry arguments like we did above in the case of parity. The answer is that we need to choose the subgroup and subalgebra to span as few irreducible representations as possible while including the $A_1$ irreducible representation. That is, we do not want to allow the state to leak into symmetry spaces other than $A_1$ if it can be avoided. Let us separately consider the case where each Pauli string transforms as an irreducible representation (and is not a sum of irreducible representations). In that case we just restrict our subgroup and subalgebra to operators that transform as the $A_1$ representation, which means they transform $A_1$ states to other $A_1$ states. That is what we did when we restricted to operators that conserve parity in the case of the H$_4$ molecule. We illustrate these considerations with two other simulations, the 10-qubit LiH problem and the 12-qubit BeH$_2$ problem. 

Figure~\ref{Fig:LiH_orbitals} shows a schematic of the LiH molecule. Here, we analyze symmetries in the context of symmetry-adapted atomic orbitals for the sake of simplicity. Although our HF reference state is always given by a Slater determinant of molecular orbitals, we can equivalently use symmetry-adapted atomic orbitals for the analysis since they possess all the same symmetries and quantum numbers as the molecular orbitals that we build from them.  We freeze two electrons in the 1s orbital of Li and include the 1s orbital of H, as well as the 2s and 2p orbitals of Li in the state space, which the remaining two electrons are allowed to occupy. 
\begin{figure}[h]
	\centering
	\includegraphics[width=0.5\textwidth]{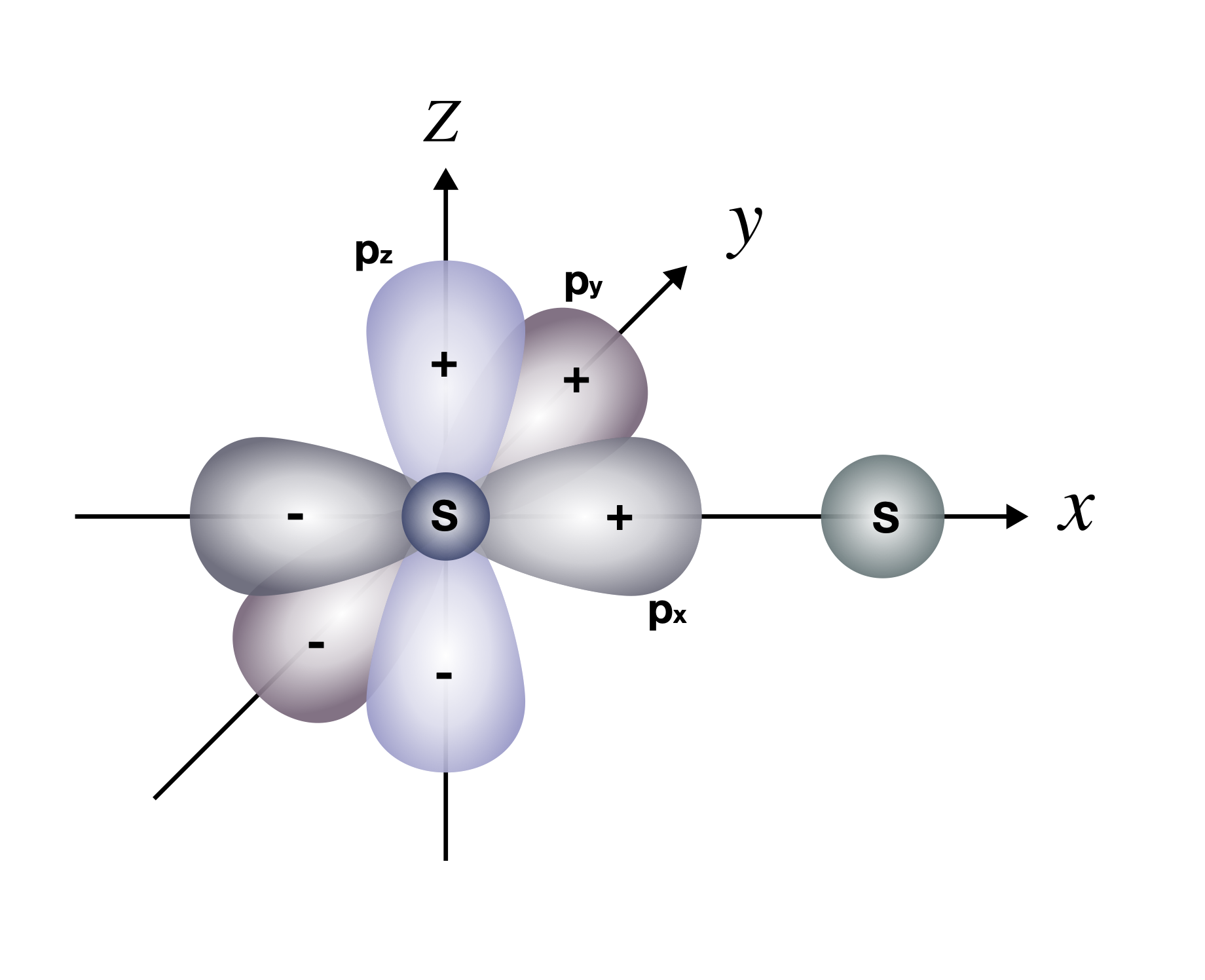}
	\caption{Schematic of the orbitals included in the ADAPT-VQE simulation of the LiH molecule. We freeze two 1$s$ electrons in the Li atom. The remaining two electrons are assumed to occupy the second shell of Li and the first shell of the H atom.}
	\label{Fig:LiH_orbitals}
\end{figure}
We do not need to consider the full symmetry group of LiH. It is enough to include those symmetry operators that allow us to discriminate between different orbitals in terms of symmetry. In our case it suffices to include reflection in the $xy$ and $xz$ planes and a $\pi$-rotation around the $x$-axis. Table \ref{tab:LiH_character_table} shows the character table of the irreducible representations of LiH. Including more operators into the symmetry group would not allow us to split the three $A_1$ orbitals in this case. 
\begin{table}[h!]
	\centering
	\begin{tabular}{ |c|c|c|c|c| }
		\hline
		 & $Id$ & $\sigma_{xy}$ &$\sigma_{xz}$ & $R_x(180)$\\
		\hline
		$A_1$ ($s_H,s_{Li},p_x$) & $1$ &$1$ & $1$ & $ 1$\\
		\hline
		 $B$ ($p_z$)     & $1$ &$-1$ & $1$ & $-1$\\
		\hline
		$C$  ($p_y$)    & $1$ &$1$ & $-1$ & $-1$ \\
		\hline
		$D$   (not present)   & $1$ &$-1$ & $-1$ & $1$\\
		\hline
	\end{tabular}
	\caption{The character table of the LiH reduced symmetry point group (C$_{\infty \text{V}}$). Out of all the symmetries of the system, we only include the rotation by $180^{\text{o}}$ around the x-axis and two reflection planes ($xy$ and $xz$).} 
	\label{tab:LiH_character_table}
\end{table}
In the case of LiH, each Pauli string transforms as an irreducible representation, and not as a direct sum of irreducible representations. To see why this is the case, one has to notice that each classical state of the simulator represents a state of the system that belongs to a representation defined as a direct product of irreducible representations of occupied orbitals. In our case each such irreducible representation is one-dimensional, as shown in Table~\ref{tab:LiH_character_table}, so the product will also be a one-dimensional irreducible representation. Thus, we conclude that each classical state transforms according to a one-dimensional irreducible representation. Now a Pauli string always transforms one classical state to another (not to a superposition state). It follows from the above then that a state belonging to a particular irreducible representation will still belong to a single irreducible representation after a Pauli string acts on it. This in turn can only be the case if each Pauli string transforms as a one-dimensional irreducible representation. As mentioned above, in this case we can simply restrict the operators in the pool to those transforming as the $A_1$ representation.  In other words, these operators will always keep the state inside of the $A_1$ representation, meaning that they will change the occupation of $p_z$  and $p_y$ orbitals by zero or two. Combined with the restrictions on the particle number and spin quantum numbers (condition 1 above), the rules described above allow us to construct a pool for the LiH problem (here starters are distinquished with bold test): 
\begin{equation}
	\begin{aligned}
		&\textbf{XYYZIIZIZY}, \textbf{XYYYIZZZII}, \textbf{YYIZZZIZXY}, \\
		&\textbf{XXZXZIIIYI},\textbf{XYZYIZZIYI}, \textbf{XXXZIIZZZY},\\
		&\textbf{XXIIYXZZII}, \textbf{YXZZIZYYII},\text{XXIYIIXYZY}, \\
		&\text{IIZIZZYYXY}, \text{ZZXZXXIIZY}, \text{YZZZXYZZZY},\\
		&\text{XYXZXXXYZY}, \text{IXIZXXZZYI}.
	\end{aligned}
	\label{Pool:LiH}
\end{equation}
In Appendix \ref{app:LiH} we consider this pool in detail and show how it obeys rules 1-5 above.
Figure~\ref{Fig:LiH_10_qubits}(a) shows the dissociation curve for the LiH molecule, modeled as a 10-qubit system. Running ADAPT-VQE with the pool given in Eq.~\eqref{Pool:LiH} allows us to converge the simulator state to the LiH ground state with an energy error less than $10^{-8}$ Ha relative to the FCI energy. Figure~\ref{Fig:LiH_10_qubits}(b,c) shows the absolute error relative to the FCI energy, as well as the ADAPT gradient, as a function of iteration number for three different bond lengths of $1.1$ $\AA$, $1.5$ $\AA$, and $2.5$ $\AA$. We see that the ADAPT-VQE convergence pattern is similar to that shown in Fig.~\ref{Fig:random_hamiltonians_error_vs_iteration} across a wide range of bond lengths near and far from equilibrium. 
\begin{figure}[h]
	\centering
	\includegraphics[width=0.445\textwidth]{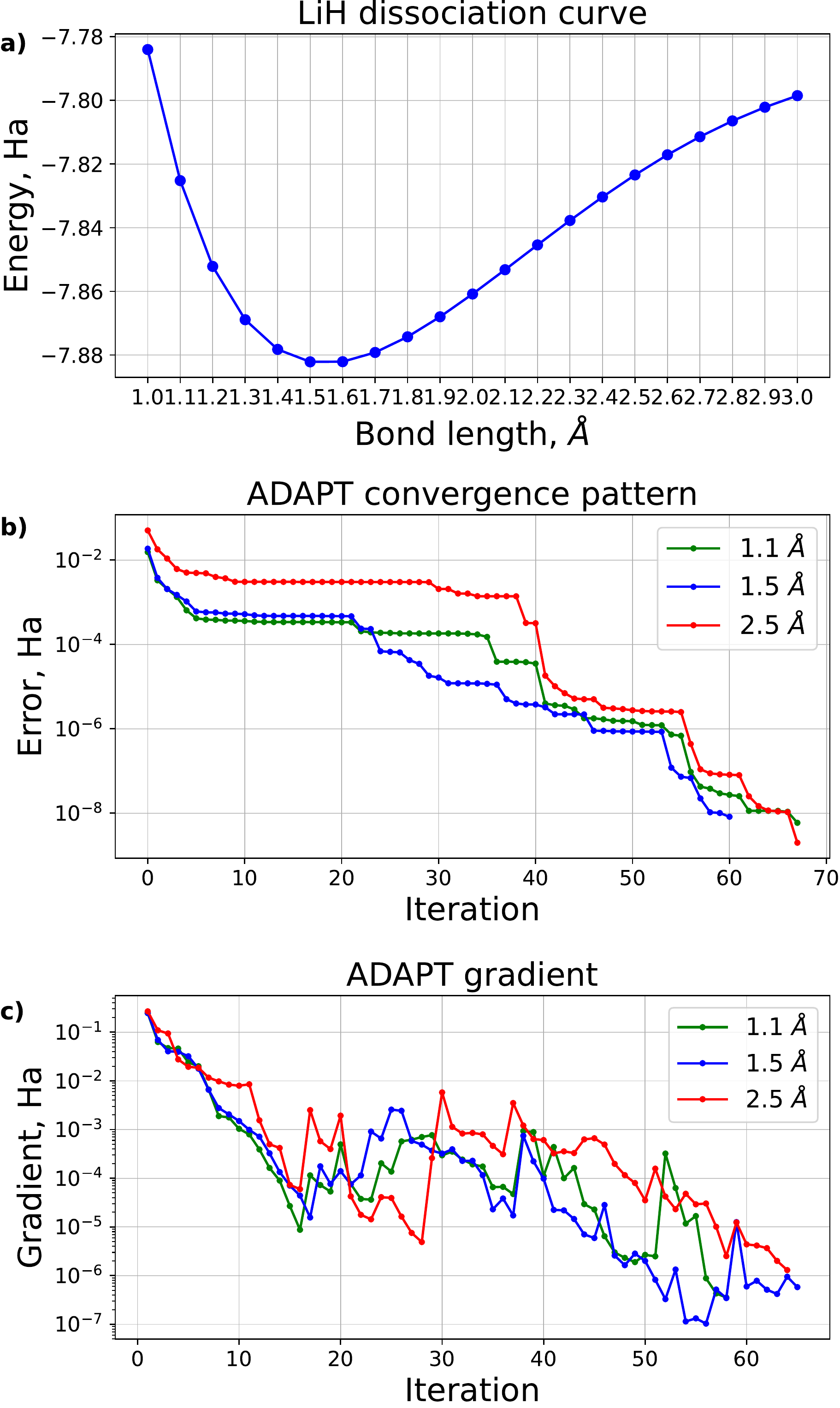}
	\caption{(a) The dissociation curve for the LiH molecule, obtained from a 10-qubit ADAPT-VQE simulation using the 14-operator symmetry-preserving pool from Eq.~\eqref{Pool:LiH}. The dots are the results of individual VQE simulations, while the curve is a guide to the eye. The absolute error in each simulation is less than $10^{-8}$ Ha relative to the FCI energy. (b) The absolute error and (c) the energy gradient vs. iteration number of the ADAPT-VQE simulation at three different bond lengths of $1.1$ $\AA$, $1.5$ $\AA$ and $2.5$ $\AA$.}
	\label{Fig:LiH_10_qubits}
\end{figure}

As our final example, we consider the BeH$_2$ molecule. Figure~\ref{Fig:BeH2_orbitals} shows a schematic of this molecule along with the orbitals relevant for the simulation. We again apply the frozen-core approximation and freeze two electrons in the 1s coordination shell of the Be atom. We restrict which atomic orbitals the electrons are allowed to occupy to two 1s orbitals of the H atoms and to the second coordination shell of Be (2s and 2p orbitals), which leaves us with 6 atomic and 12 spin orbitals in total, resulting in a 12-qubit simulation problem. We introduce the symmetry adapted orbitals:
\begin{equation}
	\begin{aligned}
		&S_+=S_{H1}+S_{H2},\\
		&S_-=S_{H1}-S_{H2}.
	\end{aligned}
	\label{Eq:symmetry_adapted_orbitals}
\end{equation}
\begin{figure}[h]
	\centering
	\includegraphics[width=0.5\textwidth]{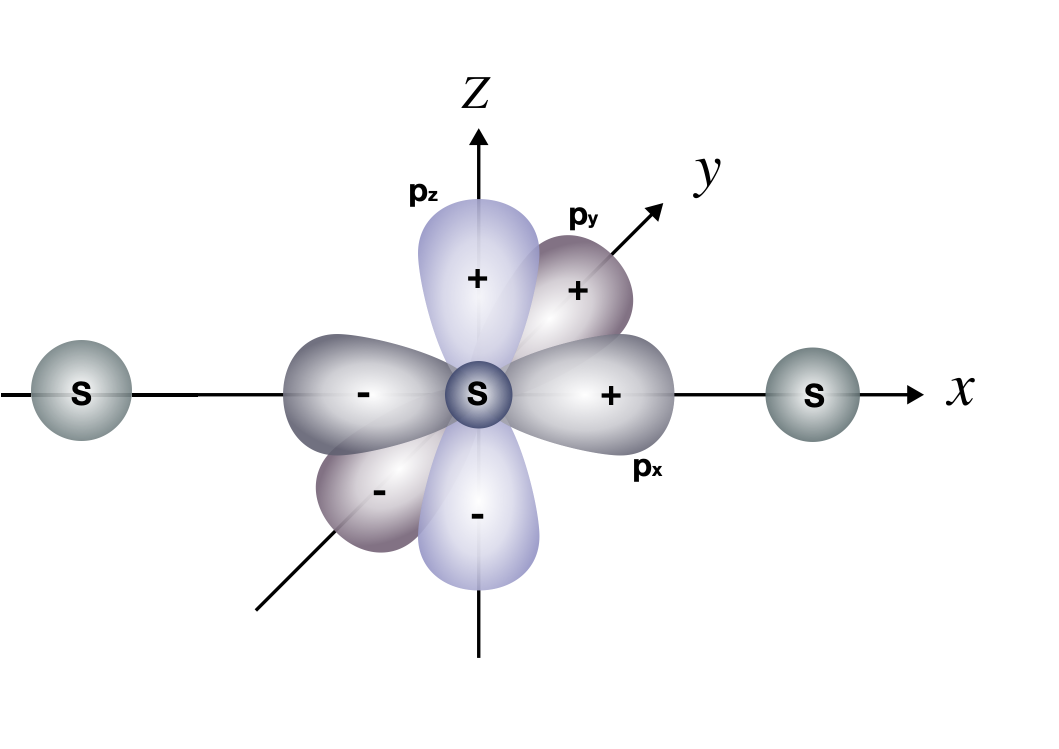}
	\caption{Schematic of the orbitals included in the VQE simulation of the BeH$_2$ molecule. We freeze two electrons in the first coordination shell of the Be atom. The remaining electrons are allowed to occupy the second shell of Be and the first shells of the two H atoms.}
	\label{Fig:BeH2_orbitals}
\end{figure}
Again, we do not need to consider the full symmetry group of the BeH$_2$ molecule. In order to discriminate between different relevant orbitals in terms of symmetry, it suffices to include the inversion, rotations by $180^{\text{o}}$ around the axes $x,y,z$ and reflection with respect to the $xy$, $xz$ and $yz$ planes.  Table~\ref{tab:BeH2_character_table} shows the character table of the corresponding symmetry group. Including more operators into the symmetry group would not allow us to separate by symmetry the two $A_1$ or two $B$ orbitals in this case. 
\begin{table}[h!]
	\centering
	\begin{adjustbox}{width=0.48\textwidth}
	\begin{tabular}{ |c|c|c|c|c| c|c|c|c|}
		\hline
		& $Id$ & I & $\sigma_{xy}$ &$\sigma_{xz}$ &$\sigma_{yz}$& $R_x$& $R_y$& $R_z$\\
		\hline
		$A_1$  ($s_+,s_{Be}$)  & $1$ &$1$ & $1$ & $ 1$ & $1$ &$1$ & $1$ & $ 1$\\
		\hline
		$B$ ($s_-,p_z$)               & $1$ &$-1$ & $-1$ & $ 1$ & $1$ &$-1$ & $-1$ & $ 1$\\
		\hline
		$C$  ($p_x$)                     & $1$ &$-1$ & $1$ & $ 1$ & $-1$ &$1$ & $-1$ & $ -1$\\
		\hline
		$D$ ($p_y$)                      & $1$ &$-1$ & $1$ & $ -1$ & $1$ &$-1$ & $1$ & $ -1$\\
		\hline
	\end{tabular}
	\end{adjustbox}
	\caption{The character table of the BeH$_2$ reduced symmetry point group. Not all irreducible representations are displayed.} 
	\label{tab:BeH2_character_table}
\end{table}
Similar to LiH, in the case of BeH$_2$ each Pauli string transforms as an irreducible representation, and not as a direct sum of irreducible representations. As mentioned above, in this case we can simply restrict the operators in the pool to those transforming like the $A_1$ representation. In other words, these operators will always keep the state within the $A_1$ representation, meaning that they will change the occupation of each of the B, C, and D representations by an even number of electrons. Combined with the restrictions on the number of particles and spin quantum numbers, the rules described above allow us to construct a pool for the BeH$_2$ problem (here starters are distinquished with bold): 
\begin{equation}
	\begin{aligned}
		&\textbf{ZYXIZZZZZYYI}, \textbf{YXIIZZIIYYII}, \textbf{ZIXYZZZIYYII}, \\
		&\textbf{XXIZZZYXIIII},\textbf{XYZIZIYYZIII}, \textbf{IIYXYYZZZZII},\\
		&\textbf{ZZYXIZYYIIII}, \textbf{YZIXZZZIIYYI},\textbf{IXXZIIIZZXYI}, \\
		&\textbf{YZXZZIZZYZYI}, \text{XXXZYXXXYXYI},\\ &\text{ZXIIIZZZZYII},
		\text{XIZZIZXYZXII}, \text{XIIIZZXYYIXY} ,\\ &\text{YZYXZIZIXZXY},
		\text{ZZZIZIIZXXXY},\\ &\text{IZZZYYYXYXXY}.
	\end{aligned}
	\label{Pool:BeH2}
\end{equation}

Figure~\ref{Fig:BeH2_12_qubits}(a) shows the dissociation curve for the BeH$_2$ molecule, modeled as a 12-qubit problem. Running ADAPT-VQE with the pool from Eq.~\eqref{Pool:BeH2} allows us to converge the simulator state to the BeH$_2$ ground state with an error less than $10^{-8}$ Ha relative to the FCI energy. Figure~\ref{Fig:BeH2_12_qubits}(b,c) shows the absolute error relative to the FCI energy, as well as the ADAPT gradient, at three different bond lengths of $1.0$ $\AA$, $1.3$ $\AA$, and $1.8$ $\AA$. These results again show the importance of building symmetries into the operator pool. 
\begin{figure}[h!]
	\centering
	\includegraphics[width=0.45\textwidth]{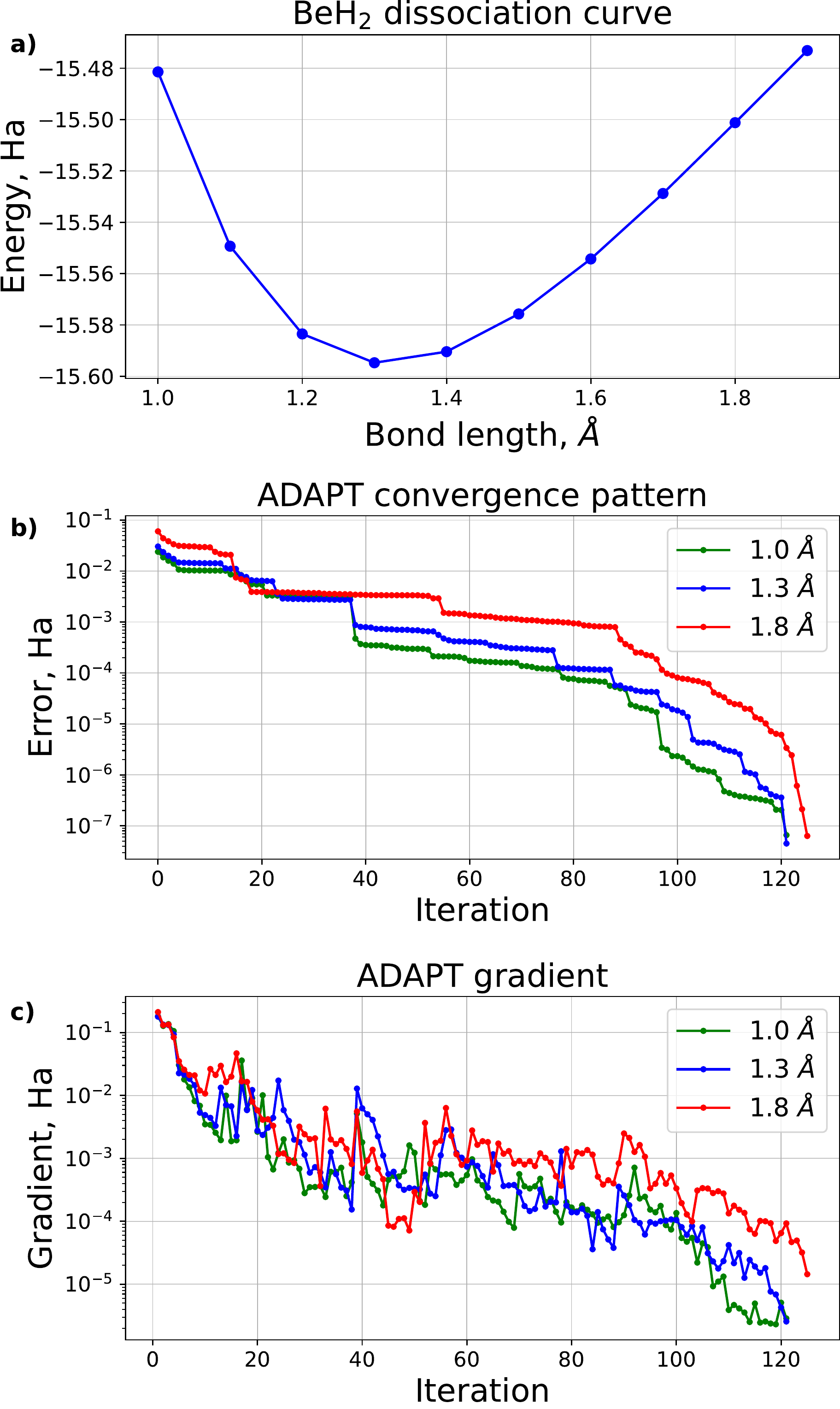}
	\caption{(a) The dissociation curve for the BeH$_2$ molecule, obtained from a 12-qubit ADAPT-VQE simulation using the 17-operator symmetry-preserving pool from Eq.~\eqref{Pool:BeH2}. The dots are the results of individual VQE simulations, while the curve is a guide to the eye. The absolute error in each simulation is less than $10^{-8}$ Ha relative to the FCI energy. (b) The absolute error and (c) the energy gradient vs. iteration number of the ADAPT-VQE simulation at three different bond lengths of $1.0$ $\AA$, $1.3$ $\AA$ and $1.8$ $\AA$.}
	\label{Fig:BeH2_12_qubits}
\end{figure}
\section{Conclusion}
In this work, we showed how to minimize the number of measurements that need to be performed during adaptive variational quantum eigensolver algorithms. We did this by establishing a general set of criteria that determine when a given Pauli operator pool is capable of exactly representing the true ground state of the system being simulated, and we proved that the minimal size of such pools scales linearly in the number $n$ of qubits. This finding means that we can reduce the measurement cost from $\mathcal{O}(n^8)$ to  $\mathcal{O}(n^5)$ in the case of molecular problems. The general criteria we introduced also allow us to systematically construct these minimal complete pools. We further showed that, when the simulated system possesses symmetries, we must take care to incorporate these into the pool to avoid algorithmic convergence issues. This can also be done systematically by leveraging group theoretic techniques. We demonstrated the utility of our approach by finding explicit, symmetry-adapted operator pools for several different molecules and showed that when these are used in the qubit-ADAPT-VQE algorithm, rapid convergence to the exact ground state can be achieved. Combining qubit-ADAPT-VQE with symmetry-adapted minimal complete pools allows one to minimize both circuit depths and measurement counts, bringing the quantum simulation of practical, classically intractable problems on near-term devices closer to fruition.

\section*{Acknowledgments}

This work is supported by the Department of Energy, grant no. DE-SC0019199, and the National Science Foundation, grant no. 1839136.

\setcounter{theorem}{0}
\appendix
\appendixpage
\addappheadtotoc

\section{Minimal complete pools}\label{app:Minimal complete pools}
\subsection*{Preliminaries and Notation}
In this appendix we discuss the generic properties of minimal complete pools (MCPs). We investigate the properties of operator pools, that are defined as sets of Pauli strings $\{\hat{P}_i\}$, each of which is capable of generating a parametrized unitary $\exp(\alpha\hat{P}_i)$. We call an operator pool complete if for any two real states $\ket{\psi}$ and $\ket{\phi}$ the product of these unitaries can transform one to the other:
\begin{eqnarray}\label{Eq:completeness_definition}
\ket{\psi}=\prod_i\exp(\alpha_i\hat{P}_i)\ket{\phi}.
\end{eqnarray}
We call a complete pool minimal if there is no complete pool of smaller size. In this section we show how to identify all  minimal complete pools and discuss their algebraic properties. In order to not work with imaginary matrix entries, we will use $iY$ instead of ordinary Pauli $Y$. Still, for conciseness we will omit writing the factor $i$ in all expressions. The Pauli strings containing odd numbers of $Y$ operators will be referred to as odd Pauli strings ($O$ operators). In contrast, the Pauli strings with even numbers of $Y$ operators will be referred to as even ($E$ operators). The following multiplication rules apply for Pauli operators in this notation
\begin{eqnarray}\label{product rules}
	\begin{aligned}
		&Y\cdot X=-X\cdot Y=Z,\\
		&Z\cdot X=-X\cdot Z=Y,\\
		&Z\cdot Y=-Y\cdot Z=X.
	\end{aligned}
\end{eqnarray}
These rules indicate that any two Pauli strings, $A$ and $B$, either commute or anticommute. Table~\ref{tab:OE_multiplication_rules} lists the parity of their product, depending on whether $A$ and $B$ are odd ($O$) or even ($E$).
\begin{table}[h!]
	\centering
	\begin{tabular}{ |c|c| }
		\hline
		[$A,B$]=0 &  \{$A,B$\}=0\\
		\hline
		$O\cdot O=E$  & $O\cdot O=O$ \\ 
		$E\cdot E=E$  & $E\cdot E=O$ \\ 
		$O\cdot E=O$  & $O\cdot E=E$ \\ 
		\hline
	\end{tabular}
	\caption{Parity of $A\cdot B$, depending on whether $A$ and $B$ are odd ($O$) or even ($E$) and on whether they commute or anticommute.} 
	\label{tab:OE_multiplication_rules}
\end{table}

Now let two odd Pauli strings $O_1$ and $O_2$ anticommute. Then we can perform the following similarity transformation:
\begin{equation}
	\exp(\frac{\pi}{4}O_1)O_2\exp(-\frac{\pi}{4}O_1)=\frac{1}{2}[O_1,O_2].
	\label{Eq: commutator_as_a_similatiry_transformation}
\end{equation}

\begin{theorem}[necessary condition of completeness]\label{theorem:necassary_condition_of_completeness}
Let odd Pauli strings $O_1$, $O_2$, \ldots, $O_k$ form a complete pool in the space of $n$ qubits. Then they generate an algebra, containing at least $2^n-1$ operators, each flipping different sets of qubits (for example for two qubits that could be $Y_1$, $Z_1Y_2$, $X_1Y_2$).
\end{theorem}
\begin{proof}
Assume we want to transform $\ket{00...0}$ to $\sqrt{1-p}\ket{00...0}+\sqrt{p}\ket{1_{i_1}1_{i_2}....1_{i_r}00....0}$. According to the completeness definition, we have to be able to do this with a finite product $\prod_i\exp(\alpha_iO_i)$, which according to the theory of Lie groups can be expressed as $\exp(\hat{g})$, where the operator $\hat g$ belongs to the Lie algebra generated by the pool $O_1$, $O_2$, \ldots, $O_k$. This Lie algebra contains odd Pauli strings only, as the commutator of two Pauli strings coincides with their product whenever it is nonzero, and the product is again odd according to Table~\ref{tab:OE_multiplication_rules}. In the limit $p\rightarrow 0$, it holds that $\hat{g}\rightarrow 0$ and thus for small $p$ we can expand $\exp(\hat{g})=1+\hat{g}$. The operator $\hat{g}$ must be of the form $\sqrt{p}\times \text{operator}$, flipping $0_{i_j}\rightarrow 1_{i_j}$ ($j\in \{1, \dots, r\}$), and so the Lie algebra must contain the Pauli strings that flip the corresponding qubits. Arbitrariness of the set $\{i_1,\ldots,i_r\}$ means the number of the strings in the complete algebra is at least $2^n-1$.
\end{proof}
\begin{remark}
Most probably this condition is also sufficient; to prove it one would have to show that the Lie algebra, which contains Pauli strings that flip any set of qubits, can transform $\ket{00....0}$ to any other state. Then one would have to show that any $\exp(\hat{g})$, where $\hat g\in$ Lie algebra, can be expressed as a $\emph{finite}$ product of $\exp(\alpha_iO_i)$. We will not prove sufficiency here, as it will be an automatic consequence of the theory later.
\end{remark}
\subsection*{Product group of Pauli strings and their Lie algebra }
Let $O_1$, $O_2$, \ldots, $O_k$ be odd Pauli strings that generate a Lie algebra through all possible commutators $[O_i,O_j]$, $[[O_i,O_j],O_m]$, $[[O_i,O_j],[O_m,O_n]]$, .... Let us also consider a product group of arbitrary products $O_iO_j...O_m$, generated by $O_1$, $O_2$, \ldots, $O_k$. Regarding this group, one important remark must be made. If $[O_1,O_2]\neq 0$, then $\{O_1,O_2\}=0$ and thus $O_1O_2=-O_2O_1$. We do not need to consider negative Pauli strings, so we will work with a factor group $G$ generated over a normal subgroup $\{I\times I\times...\times I,-I\times I\times...\times I\}$ of the total group. In this new group $G$, the elements $A$ and $-A$ are considered the same, making this group Abelian. Note that we will still refer to commutation or anticommutation of Pauli strings from this group, but in our case the products $Z\cdot Y$ and $Y\cdot Z$ correspond to the same element of $G$. Whenever we consider commutators, we will arrange it such that the result is a positive Pauli string.

 If we notice that the nonzero commutator of Pauli strings coincides with their product up to a numerical factor (irrelevant for our purposes), we arrive at the following Lemma:
\begin{lemma}\label{lemma:algebra_is_a_subet_of_group}
The Lie algebra generated by $O_1$, $O_2$, \dots, $O_k$ is a subset of odd Pauli strings from the group $G$.
\end{lemma}
\begin{remark}
	The Lie algebra must not span all odd strings from the group $G$. For example, if the generators are $Y_1$, $Y_2$, $Y_3$, then the group $G$ contains the odd string $Y_1Y_2Y_3$, but this operator cannot be obtained through commutators, and thus this element is not in the Lie algebra.
\end{remark}
  
  \begin{remark}
The important consequence of Lemma \ref{lemma:algebra_is_a_subet_of_group} is that if group $G$ is too small to contain the operators necessary for completeness, then the algebra is also too small for that.
  \end{remark}
\begin{lemma}
	Let $A_1$ and $A_2$ be two Pauli strings that flip the same qubits: $A_1\ket{00....0}=\pm A_2\ket{00....0}$. Then $A_2=P\cdot A_1$, where $P$ is an even Pauli string containing only $Z$ and $I$ operators, and iff $A_1$ and $A_2$ have the same parity, $[A_1,P]=0$.
	\label{lemma:two_paulis_flipping_same_qubits}
\end{lemma}
\begin{proof}
	Let us choose an arbitrary odd Pauli string, e.g., 
	\begin{equation}
		A_1=YXIZYXY.
	\end{equation}
	Whenever one encounters $Y$ or $X$ in $A_1$, for the operator $A_2$ they either remain invariant or transform according to $Y\rightarrow X$ and $X\rightarrow Y$. This can only be achieved by multiplying the corresponding Pauli operator by $I$ or $Z$. Whenever $I$ or $Z$ is encountered in $A_1$, they do not flip qubits and thus they must be multiplied by $I$ or $Z$ to remain like that. So $P$ consists of only $I$ and $Z$ operators. Whenever $A_1$ and $A_2$ are both odd or both even, $X$ and $Y$ are flipped an even number of times to transform $A_1$ into $A_2$. Thus it must be that $[A_1,P]=0$. (This can also be deduced from the rules in Table \ref{tab:OE_multiplication_rules}.)
\end{proof}
\subsection*{Internal structure of the group $G$}
Consider a product group $G$ generated by $O_1$, $O_2$, \ldots, $O_p$. Let us also consider all Pauli strings $H=\{P_1,\ldots,P_l\}$ from $G$, such that they consist of operators $I$ and $Z$ only. This set forms a subgroup of $G$. Because this subgroup is Abelian and each element is its own inverse, we conclude that there are only $k$ independent Pauli strings $P_1$, \ldots, $P_k$ in $H$, and all others are just products of these (of which there are $l=2^k$).
\begin{lemma}
	We say a unitary operator $\hat{U}$ performs a similarity transformation on the operator $\hat{O}$, if $\hat{O}\rightarrow\hat{U}^\dagger\hat{O}\hat{U}$. Then the $k$ independent generators $P_1$, \ldots, $P_k$ can be chosen in such a way that they can be transformed into $Z_1$, $Z_2$, \ldots, $Z_k$ using a certain similarity transformation.
\end{lemma}
\begin{proof}
	Without loss of generality, we can assume $P_1$ contains the operator $Z_1$. If it does not, then we can perform a similarity transformation with a SWAP operation (which is implemented by a real and orthogonal matrix), such that $P_1$ contains $Z_1$. Now if $P_i$ ($i\neq 1$) also contains $Z_1$, we replace $P_i\rightarrow P_iP_1$ with the other matrix from the subgroup $H$, so $P_i$ will contain $I$ instead of $Z_1$. We can now assume $P_2$ contains $I_1Z_2$; if it does not, we perform a SWAP such that $P_2$ then contains $Z_2$. Whenever $P_i$ ($i\neq 2$) contains $Z_2$, we replace $P_i\rightarrow P_2P_i$. Repeating this procedure, we will construct the generators
	\begin{eqnarray}
		\begin{aligned}
			&Z_1I_2I_3\dots I_kP_{kn}^1,\\
			&I_1Z_2I_3\dots I_kP_{kn}^2,\\
			&\quad\qquad\vdots\\
			&I_1I_2I_3\dots Z_kP_{kn}^k,\\
		\end{aligned}
	\end{eqnarray}
	where $P_{kn}^i$ is a string on the qubits $k+1$ to $n$ that contains $Z$ or $I$ operators. Let us perform a similarity transformation generated by $I_1\times I_2\times...\times Y_i\times...\times I_k\times P_{kn}^i$, according to Eq.~\eqref{Eq: commutator_as_a_similatiry_transformation}. It will only transform $P_i$ (the generator commutes with all others), and the result will be
	\begin{equation}
		I_1\times I_2\times...Z_i...\times I_k\times P_{kn}^i\rightarrow I_1\times I_2\times...X_i...\times I_n.
	\end{equation}
	Now performing a similarity transformation generated by $I_1\times I_2\times...Y_i...\times I_n$, we obtain 
	\begin{equation}
		I_1\times I_2\times...X_i...\times I_n\rightarrow I_1\times I_2\times...Z_i...\times I_n.
	\end{equation}
	All other $P_j$ strings are left intact by this transformation. Performing this procedure for every $P_i$, we bring the generators $P_1, \ldots, P_k$ to $Z_1$, $Z_2$, \ldots, $Z_k$.
\end{proof}

Without loss of generality, we can now assume the subgroup $H$ of $G$ to be generated by $Z_1$, $Z_2$, \ldots, $Z_k$. The subgroup $H$ is normal due to the Abelian property of $G$, and so we can consider the factor group $\tilde{G}=G/H$. Clearly, each element of $\tilde{G}$ is a set of Pauli strings (odd and even), that flip the exact same qubits (act upon $\ket{00...0}$ in exactly the same way). No two elements of the factor group $\tilde{G}$ act in the same way on $\ket{00....0}$, as otherwise we would have more than $k$ independent $Z$-Pauli strings (see Lemma~\ref{lemma:two_paulis_flipping_same_qubits}). The factor group $\tilde{G}$ thus contains not more than $2^n$ elements. But according to the necessary condition of completeness (Theorem \ref{theorem:necassary_condition_of_completeness}), we need operators that flip all possible combinations of qubits, so the number of elements in the factor group for a complete pool is $2^n$. Because again each element of $\tilde{G}$ is its own inverse and the group is Abelian, $\tilde{G}$ can be generated with $n$ additional Pauli strings (multiplied by $H$), that without loss of generality can be chosen to be $Y_1P_{kn}^1$, $Y_2P_{kn}^2$,...,$Y_kP_{kn}^n$, $Y_{k+1}Q_{kn}^{k+1}$,...,$Y_{n}Q_{kn}^{n}$. Here $P_{kn}^i$ is defined as above, while $Q_{kn}^i$ is a $Z$-Pauli string on qubits $k+1$ to $n$, but not on the $i$th one (this one is reserved for the $Y$ operator in our case). When multiplied by $H$, these generators will recover all the elements of $G$. Thus we have shown that the group $G$, generated by a complete pool, contains $n+k$ generators. But $k$ cannot be arbitrary, because products of generators containing $Y$ operators are sometimes even. If such an even operator commutes with all $Z_1$, $Z_2$,...., $Z_k$, then according to the rules in Table~\ref{tab:OE_multiplication_rules}, it will create an element of $\tilde{G}$ consisting of even Pauli strings only, and according to Theorem 1 the initial generators do not form a complete pool. This never happens for operators containing $Y_1$, $Y_2$, ..., $Y_k$, so we only have to make sure $Y_{k+1}Q_{kn}^{k+1}$, $Y_{k+2}Q_{kn}^{k+2}$, ..., $Y_{n}Q_{kn}^{n}$ do not generate even Pauli strings. First of all, this means all of them mutually anticommute. Otherwise the product of two commuting odd Pauli strings would generate an even Pauli string, commuting with $Z_1$, $Z_2$, ..., $Z_k$ (see the rules in Table~\ref{tab:OE_multiplication_rules}). Second, there are not more than two operators ($k\geq n-2$). Indeed, for three mutually anticommuting strings $a,b,c$, we would necessarily have $[[a,b],c]=0$, and then $abc$ would be an even string that commutes with $Z_1$, $Z_2$,...., $Z_k$. We  thus conclude that $k\geq n-2$, and so the minimal number of generators capable of generating a complete pool is $n+n-2=2n-2$. We have thus proved the following lemma:
\begin{lemma}
The most general form of the generators of the minimal group $G$ up to a similarity transformation is $Z_1$, $Z_2$, ..., $Z_{n-2}$, $Y_1P_{n-2,n}^1$, $Y_2P_{n-2,n}^2$, ..., $Y_{n-2}P_{n-2,n}^n$, $Y_{n-1}$, $Z_{n-1}Y_n$. The last two qubits can always be switched to $Y_{n-1}Z_n$, $Y_n$ using a swap operation (these two options are the only possible ones).
\end{lemma}
It turns out we can simplify the whole generator set even further:
\begin{lemma}
Using similarity transformations, $Z_1$, $Z_2$, ..., $Z_{n-2}$, $Y_1P_{n-2,n}^1$, $Y_2P_{n-2,n}^2$, ..., $Y_{n-2}P_{n-2,n}^n$, $Y_{n-1}$, $Z_{n-1}Y_n$ can be transformed into $Z_1$, $Z_2$, ..., $Z_{n-2}$, $Y_1$, $Y_2$, ...,$Y_{n-2}$, $Y_{n-1}$, $Z_{n-1}Y_n$.
\end{lemma}
 \begin{proof}
We will prove this by presenting the protocol.
\begin{itemize}
	\item Step 1. Apply a similarity transformation generated by $Z_{n-1}Y_n\times Z_{i_1}Z_{i_2}....Z_{i_p}$ according to Eq.~\eqref{Eq: commutator_as_a_similatiry_transformation}. The operator $Z_i$ is present in this string if $P_{n-2,n}^i=Z_{n-1}I_n$ or if $P_{n-2,n}^i=I_{n-1}Z_n$ and is absent otherwise. \\
	\item Step 2. Apply a similarity transformation generated by $Y_n\times Z_{i_1}Z_{i_2}....Z_{i_p}$. The operator $Z_i$ is present in this string if $P_{n-2,n}^i=Z_{n-1}I_n$ or if $P_{n-2,n}^i=I_{n-1}Z_n$ and is absent otherwise. 
	These two steps remove $Z_{n-1}$ from all $P_{n-2,n}^i$ that initially contain this operator.\\ 
	\item Step 3. The first two steps scramble $Y_{n-1}$ into $X_{n-1}Y_nZ_{i_1}Z_{i_2}...Z_{i_p}$. Using other generators we can return it back to $Y_{n-1}$.\\
	\item Step 4. Apply a similarity transformation generated by $Y_{n-1}Z_nZ_{i_1}...Z_{i_p}$. The operator $Z_i$ is present if $P_{n-2,n}^i$ contains $Z_n$ and is absent otherwise. \\
	\item Step 5. Apply a similarity transformation generated by $Y_{n-1}Z_{i_1}...Z_{i_p}$. The operator $Z_i$ is present in this string if $P_{n-2,n}^i$ contains $Z_n$ and is absent otherwise. Steps 4 and 5 together remove $Z_n$ from each $P_{n-2,n}^i$ that initially contains it. \\
	\item Step 6. The two previous steps scramble the generator $Z_{n-1}Y_n$ into $X_{n-1}Y_nZ_{i_1}Z_{i_2}...Z_{i_p}$. Using other generators, we return it back to $Z_{n-1}Y_n$. \\
\end{itemize}
 \end{proof}

Now we are ready to formulate a theorem that summarizes several previous lemmas and coincides with Theorem~\ref{mainText:theorem:canonical_form_of_minimal_complete_pools} from the main text:
\begin{theorem}\label{theorem:canonical_form_of_minimal_complete_pools}
	A minimal complete pool must contain $2n-2$ Pauli string generators $O_1$,  $O_2$, ..., $O_{2n-2}$ and must generate a product group, which, up to a similarity transformation, coincides with the one generated by $Z_1$, $Z_2$, ..., $Z_{n-2}$, $Y_1$, $Y_2$, ..., $Y_{n-2}$, $Y_{n-1}$, $Z_{n-1}Y_n$. The corresponding algebra is a subset of odd strings from this group.
\end{theorem}

From now on whenever we refer to an MCP, we will always assume it generates the group in the canonical form above.
\begin{lemma}
The number of odd Pauli strings in the group $G$ is exactly $\frac{2^{n-1}(2^{n-1}+1)}{2}$. 
\label{lemma:number_of_odd_strings_in_complete_group}
\end{lemma}
\begin{proof}
	We define the binomial coefficient $C_n^k=\frac{n!}{k!(n-k)!}$. This coefficient represents the number of ways we can choose $k$ out of $n$ elements. 
	
	Let us first compute the number of all odd Pauli strings, generated by $Z_1$, $Z_2$, ..., $Z_{n-2}$, $Y_1$, $Y_2$, ..., $Y_{n-2}$. Let us consider all such strings that flip certain sets of $2k+1$ qubits. As an example, we can take $Y_1Y_2...Y_{2k+1}$. We can transform an even number of $Y$ operators into $X$ in this string using $Z_1$, ..., $Z_{2k+1}$ operators. The number of strings we obtain in this way is
	\begin{equation}
		C_{2k+1}^0+C_{2k+1}^2+...+C_{2k+1}^{2k}=2^{2k}.
	\end{equation} 
	Each such string can also be multiplied by any sequence of $Z_{2k+2}$, ..., $Z_{n-2}$, of which there are $2^{n-2-(2k+1)}$. So the total number of odd strings that flip qubits $1,2,...,2k+1$ is 
	\begin{equation}
		2^{2k}\times2^{n-2-(2k+1)}=2^{n-3}.
	\end{equation}
	The total number of strings that flip all possible sets of $2k+1$ qubits is
	\begin{equation}
		C_{n-2}^{2k+1}2^{n-3}.
	\end{equation}
	If we consider all odd strings that flip $2k$ qubits, their number is
	\begin{equation}
		\begin{aligned}
			&C_{n-2}^{2k}(C_{2k}^{1}+C_{2k}^3+...+C_{2k}^{2k-1})\times2^{n-2-2k}=\\
			&C_{n-2}^{2k}\times2^{2k-1}\times2^{n-2-2k}=C_{n-2}^{2k}\times2^{n-3}.
		\end{aligned}
	\end{equation}
	So the number of all odd strings generated by $Z_1$, $Z_2$, ..., $Z_{n-2}$, $Y_1$, $Y_2$, ..., $Y_{n-2}$ is
	\begin{equation}
		\begin{aligned}
			\text{Odd}=&(C_{n-2}^{1}+C_{n-2}^{2}+...+C_{n-2}^{n-2})\times2^{n-3}\\
			                       =&(2^{n-2}-1)\times2^{n-3}.\\
		\end{aligned}
	\end{equation}
	The number of even strings generated by $Z_1$, $Z_2$, ..., $Z_{n-2}$, $Y_1$, $Y_2$, ..., $Y_{n-2}$ is $2^{2n-4}$ minus the number of all odd strings, which amounts to 
	\begin{equation}
		\text{Even}=(2^{n-2}+1)\times2^{n-3}.
		\label{Eq:even_strings}
	\end{equation}
	If we now add two more generators, $Y_{n-1}$ and $Z_{n-1}Y_n$, (they generate the third operator $X_{n-1}Y_n$), then they generate new odd strings when multiplied by the even strings from Eq.~\eqref{Eq:even_strings}, and so the total number of odd strings becomes
	\begin{equation}
		\text{Odd}+3\times\text{Even}=\frac{2^{n-1}(2^{n-1}+1)}{2}.
	\end{equation}
\end{proof}

Up to this point, we only considered the size and structure of the product group an MCP would need to generate. This is a necessary condition of completeness, but at this point we still have not shown that pools satisfying this condition exist. Such pools do indeed exist. In Appendix B of Ref.~\cite{Tang_PRXQuantum2021}, a pool of size $2n-2$ was constructed explicitly and shown to transform any real state into any other state. One can check that up to renumbering of the qubits, this pool generates the group in Theorem~\ref{theorem:canonical_form_of_minimal_complete_pools}.

If a pool generates an algebra that spans all odd strings from the group in Theorem~\ref{theorem:canonical_form_of_minimal_complete_pools}, it will generate exactly the same algebra as the pool from Appendix B of \cite{Tang_PRXQuantum2021}, thus proving that this pool is also complete. This is a sufficient condition of completeness that we formulate in the following theorem: 
\begin{theorem}[sufficient condition of completeness]
 If a pool $O_1$, $O_2$, ..., $O_{2n-2}$ generates an algebra that spans all odd Pauli strings from the group in Theorem~\ref{theorem:canonical_form_of_minimal_complete_pools}, then this pool is minimal and complete. 
 \label{theorem: algebra sufficient condition}
\end{theorem}
 Therefore the pools that span all odd strings from the group $G$ form a class of complete pools that are unique up to similarity transformations.
 
 We now formulate one more necessary condition of completeness that is extremely useful when conducting searches for complete pools:
\begin{theorem}[necessary condition of completeness]
Let a pool of $2n-2$ Pauli strings $O_1$, $O_2$, ..., $O_{2n-2}$ generate the product group $G$ from Theorem~\ref{theorem:canonical_form_of_minimal_complete_pools}. If we can split the pool into two sets of operators \{$A_1$, ..., $A_k$\} and \{$B_1$, ..., $B_{2n-2-k}$\}, such that $[A_i,B_j]=0$ for any $i$ and $j$, then the pool is incomplete. 
\label{theorem: inseparability_criterion}
\end{theorem}

\begin{proof}
Let us assume the opposite is true and the pool $O_1$, $O_2$, ..., $O_{2n-2}$ is complete. Each of the new sets $\{A\}=\{A_1, ..., A_k\}$ and $\{B\}=\{B_1, ..., B_{2n-2-k}\}$ generate subgroups $A$ and $B$ of the group $G$. Because the number of generators of each of these groups is smaller than $2n-2$, we know that the set ${A}$ is incomplete and the set ${B}$ is incomplete, as they fail to fulfill the necessary condition of completeness (Theorem~\ref{theorem:necassary_condition_of_completeness}), i.e. generating $2^n-1$ different odd flippings in each of the groups $A$ and $B$. Now let us assume that some odd flipping is not part of the group $A$, but the corresponding even flipping $E_A$ is. (For example the odd string that flips the first qubit only $Y_1Z_iZ_j...$ is not in $A$, but the even string $X_1Z_kZ_m...$ is there.) Then necessarily the group $B$ must contain the corresponding odd flipping $O_B$ in order for the original pool to be complete (the algebra is a subset of the groups $A$ and $B$). But in that case, according to Lemma 2, $E_A$ and $O_B$ anticommute. This cannot be, as all elements in $A$ must commute with all elements in $B$ due to the corresponding property of the generators. This contradiction shows that if an odd flipping is not part of $A$, then the corresponding even flipping is also not there, and thus the flipping (call it $F_1$) is completely absent in $A$. Then this flipping must be present in $B$ with corresponding operator $O_B$. Analogously, the group $B$ does not contain a certain flipping (call it $F_2$), but it must be contained in $A$, and we call the corresponding operator $O_A$. Let us now consider the new flipping ($F$), implemented by the product $O_AO_B$. This flipping can be part of neither $A$ nor $B$, as in combination with $O_A$ it would generate the flipping $F_1$ and in combination with $O_B$ it would generate $F_2$, but these flippings must not be present in the corresponding groups $A$ and $B$. Thus the corresponding flipping ($F$) is absent in $A$ and $B$, and the necessary condition of completeness cannot be fulfilled for the original pool. Thus our assumption that the original pool $O_1$, $O_2$, ..., $O_{2n-2}$ is complete is wrong, and it must instead be incomplete.
\end{proof}

We now formulate a theorem that is extremely useful when searching for minimal complete pools numerically and that coincides with the completeness criteria (Theorem~\ref{mainText:theorem:completeness_criterion}) from the main text.
\begin{theorem}[completeness criterion]
	Let a pool of $2n-2$ Pauli string generators $O_1$, $O_2$, ..., $O_{2n-2}$ generate a product group $G$ as defined in Theorem~\ref{theorem:canonical_form_of_minimal_complete_pools}. The following statements are equivalent:
	\begin{itemize}
		\item (a) The pool $O_1$, $O_2$, ..., $O_{2n-2}$ is complete.
		\item (b) The pool $O_1$, $O_2$, ..., $O_{2n-2}$ cannot be split into two mutually commuting sets.
		\item (c) The algebra generated by $O_1$, $O_2$, ..., $O_{2n-2}$ spans all odd strings from the group $G$.\\
	\end{itemize}
\label{theorem:completeness_criterion}
\end{theorem}

\begin{proof}
	If (a) is true, then (b) must be true according to Theorem~\ref{theorem: inseparability_criterion}.	If (c) is true then (a) must be true according to Theorem~\ref{theorem: algebra sufficient condition}.
	Proving that (c) follows from (b) turns out to be very challenging. For now we do not have an analytical proof, but all our numerical calculations confirm this is true, and we use this statement in practice.
\end{proof}
\begin{remark}
Theorem~\ref{theorem:completeness_criterion} is a criterion that is proved half analytically and half numerically. The condition (c) can safely be used to search for complete pools, as its applicability has been proven analytically. At the same time, computing the Lie algebra for a given pool is very resource-intensive. If one needs to check many pools for completeness and select one based on some other criterion, it might take too long to do that. This is why condition (b) is extremely useful, as its computational complexity scales polynomially with the size of the pool, and thus it is a lot easier and faster to use.
\end{remark}

\section{Preparation of the pool for molecular simulations LiH and BeH$_2$}\label{app:PoolConstruction}
In this appendix we analyze the pools used to simulate the LiH and BeH$_2$ molecules. We already identified the five conditions each pool must fulfill:
\begin{enumerate}\label{enum:symmetry_constraints_on the pool}
    \item The number of electrons with a given spin changes by a multiple of 2 by any Pauli from the pool. So there is an even number of $X$, $Y$ operators acting on $\alpha$ orbitals and there is an even number of $X$, $Y$ operators acting on $\beta$ orbitals.
    \item Each Pauli operator must transform as the $A_1$ irreducible representation of the symmetry group. This condition means that the Pauli strings do not change the symmetry of the states they act upon. We note that for more complicated molecules such Pauli strings might not exist, in which case this condition should be replaced by the requirement that the leakage of states from the $A_1$ symmetry space must be minimal.
    \item The pool must contain enough starters for ADAPT-VQE to start. A starter is a Pauli string that conserves the symmetries of the Hartree-Fock state and contains exactly 4 $X$, $Y$ operators.
    \item The pool generates the biggest subgroup and subalgebra of those generated by an MCP, that contain Pauli strings obeying conditions 1-2.
\end{enumerate}

Conditions 3 and 4 here deserve a special comment. We do not have an exact answer on how many starters must be in the pool to fulfill condition 3. For our simple simulations, it was sufficient to have around half of the operators as starters. We note that one cannot take all operators in the pool to be starters. In that case the generated subgroup and subalgebra will not contain operators that destroy an odd number of electrons from the Hartree-Fock state and recreate them in other orbitals. We note however that a pool can contain all but one operator as starters. We gave this example already in the pool we used to simulate the H$_4$ molecule (Eq.~\eqref{Pool:H4}). 

For the cases of LiH and BeH$_2$, it is still possible to check condition 4 and make sure the pool generated the right subgroup and subalgebra (the biggest ones possible that obey conditions 1-2 below). For bigger molecules of course that might be challenging, so one might only check the group and inseparability condition instead.

Below we analyze the pools used to simulate the LiH and BeH$_2$ molecules and show that the operators in the pool obey conditions 1-2.
\subsection{Pool for LiH 10-qubit problem}\label{app:LiH}
To simulate the LiH molecule we used the following pool (see \cite{source_code_LiH} for details on the code used to generate the pool):
\begin{equation}
	\begin{aligned}
		&\text{XYYZIIZIZY}, \text{XYYYIZZZII}, \text{YYIZZZIZXY}, \\
		&\text{XXZXZIIIYI},\text{XYZYIZZIYI}, \text{XXXZIIZZZY},\\
		&\text{XXIIYXZZII}, \text{XYXZXXXYZY},\text{XXIYIIXYZY}, \\
		&\text{IIZIZZYYXY}, \text{ZZXZXXIIZY}, \text{YZZZXYZZZY},\\
		&\text{YXZZIZYYII}, \text{IXIZXXZZYI}.
	\end{aligned}
	\label{app:Pool:LiH}
\end{equation}

 In order to describe the pool in terms of symmetry, we must first put symmetry and spin labels on each of the 10 spin-orbitals. The fact that we use Hartree-Fock as an initial state means the single-particle orbitals are hybridized and can no longer belong to a particular atom as in Table~\ref{tab:LiH_character_table}. Still, single-particle excitations only mix states of the same symmetry, so we only have three single-qubit orbitals belonging to $A_1$ irreducible representation. Instead of marking the orbital with an atomic notation ($s$, $p$,...), we will label each orbital in terms of the irreducible representation it belongs to from Table \ref{tab:LiH_character_table}. Then our orbitals are represented on the simulator in the following way:
\begin{equation}
    A_1^\alpha A_1^\beta A_1^\alpha A_1^\beta C^\alpha C^\beta B^\alpha B^\beta A_1^\alpha A_1^\beta,
    \label{eq:labeling of the orbitals for LiH}
\end{equation}
while the Hartree-Fock state in the frozen core approximation looks like
\begin{equation}
    \ket{1100000000},
    \label{eq:Hartree-Fock_LiH} 
\end{equation}
where the indices $\alpha,\beta$ stand for the $1/2$ and $-1/2$ spin projections, respectively. Let us now analyze each operator from the pool in Eq.~\eqref{app:Pool:LiH} and see how it obeys the symmetry conditions we identified.
\begin{itemize}[noitemsep]
     \item  $\text{XYYZIIZIZY}$$-$starter
     
     When acting upon the Hartree-Fock state, this Pauli string destroys electrons in the first $A_1^\alpha$ and $A_1^\beta$ orbitals and creates them in the second $A_1^\alpha$ and third $A_1^\beta$ orbitals respectively. The total spin projection and number of particles of the Hartree-Fock state do not change. It also remains in the $A_1$ symmmetry subspace. That means this Pauli string conserves the symmetries of the Hartree-Fock state and at the same time contains exactly 4 $X$, $Y$ operators, so it is a starter. \\
     \item	$\text{XYYYIZZZII}$$-$starter
     
     This Pauli string conserves the symmetries of the Hartree-Fock state and at the same time contains exactly four $X$, $Y$ operators, so it is a starter.\\
     \item	$\text{YYIZZZIZXY}$$-$starter, \\
     
          This Pauli string conserves the symmetries of the Hartree-Fock state and at the same time contains exactly four $X$, $Y$ operators, so it is a starter.\\
	 \item	$\text{XXZXZIIIYI}$$-$starter, \\
	 
	      This Pauli string conserves the symmetries of the Hartree-Fock state and at the same time contains exactly four $X$, $Y$ operators, so it is a starter.\\
	 \item	$\text{XYZYIZZIYI}$$-$starter, \\
	 
	      This Pauli string conserves the symmetries of the Hartree-Fock state and at the same time contains exactly four $X$, $Y$ operators, so it is a starter.\\
	 \item	$\text{XXXZIIZZZY}$$-$starter, \\
	 
	      This Pauli string conserves the symmetries of the Hartree-Fock state and at the same time contains exactly four $X$, $Y$ operators, so it is a starter.\\
	 \item	$\text{XXIIYXZZII}$$-$starter, \\
	 
	      This Pauli will create two electrons in the $C$ orbital when it acts on the Hartree-Fock state. The resulting state will still transform as the $A_1$ irreducible representation. That means this Pauli string conserves the symmetries of the Hartree-Fock state and at the same time contains exactly four $X$, $Y$ operators, so it is a starter.\\
	 \item	$\text{XYXZXXXYZY}$$-$not a starter, \\
	 
	 This Pauli string contains more than 8 $X$, $Y$ operators, so it cannot be a starter. But it obeys rules 1-2.\\
	 \item	$\text{XXIYIIXYZY}$$-$not a starter, \\
	 	
 	 This Pauli string contains more than 6 $X$, $Y$ operators, so it cannot be a starter. But it obeys rules 1-2.\\
	 \item	$\text{IIZIZZYYXY}$$-$not a starter, \\
	 
  	 This Pauli does not preserve the number of particles in the Hartree-Fock state, so it cannot be a starter. But it obeys rules 1-2.\\
	 \item	$\text{ZZXZXXIIZY}$$-$not a starter, \\

  	 This Pauli does not preserve the number of particles in the Hartree-Fock state, so it cannot be a starter. But it obeys rules 1-2.\\	 
	 \item	$\text{YZZZXYZZZY}$$-$not a starter, \\

  	 This Pauli does not preserve the number of particles in the Hartree-Fock state, so it cannot be a starter. But it obeys rules 1-2.\\	 
	 \item	$\text{YXZZIZYYII}$$-$starter, \\

      This Pauli string conserves the symmetries of the Hartree-Fock state and at the same time contains exactly four $X$, $Y$ operators, so it is a starter.\\	 
	 \item	$\text{IXIZXXZZYI}$$-$not a starter. \\
	 
  	 This Pauli does not preserve the number of particles in the Hartree-Fock state, so it cannot be a starter. But it obeys rules 1-2.\\
	 
\end{itemize}

\subsection{Pool for BeH$_2$ 12-qubit problem}\label{app:BeH2}
To simulate the BeH$_2$ molecule we used the following pool (see \cite{source_code_BeH_2} for details on the code used to generate the pool):
\begin{equation}
	\begin{aligned}
		&\text{ZYXIZZZZZYYI}, \text{YXIIZZIIYYII}, \text{ZIXYZZZIYYII}, \\
		&\text{XXIZZZYXIIII},\text{XYZIZIYYZIII}, \text{IIYXYYZZZZII},\\
		&\text{ZZYXIZYYIIII}, \text{YZIXZZZIIYYI},\text{IXXZIIIZZXYI}, \\
		&\text{XIIIZZXYYIXY}, \text{XXXZYXXXYXYI}, \text{ZXIIIZZZZYII},\\
		&\text{XIZZIZXYZXII}, \text{YZXZZIZZYZYI}, \text{YZYXZIZIXZXY},\\
		&\text{ZZZIZIIZXXXY}, \text{IZZZYYYXYXXY}.
	\end{aligned}
	\label{app:Pool:BeH2}
\end{equation}

 In order to describe the pool in terms of symmetry, we must first put symmetry and spin labels on each of the 12 spin-orbitals. We will again label the orbitals in terms of irreducible representations from Table~\ref{tab:BeH2_character_table}:
\begin{equation}
    A_1^\alpha A_1^\beta B^\alpha B^\beta C^\alpha C^\beta D^\alpha D^\beta A_1^\alpha A_1^\beta B^\alpha B^\beta,
    \label{eq:labeling of the orbitals for BeH2}
\end{equation}
while the Hartree-Fock state in the frozen-core approximation looks like
\begin{equation}
    \ket{111100000000}.
    \label{eq:Hartree-Fock_LiH} 
\end{equation}
Let us now analyze each operator from the pool in Eq.~\eqref{app:Pool:BeH2} and see how it obeys the symmetry conditions we identified.
\begin{itemize}[noitemsep]
 	\item $\text{ZYXIZZZZZYYI}$$-$starter, \\
 	
 	The total spin projection and number of particles of the Hartree-Fock state do not change. It also remains in the $A_1$ symmetry subspace. That means this Pauli string conserves the symmetries of the Hartree-Fock state and at the same time contains exactly 4 $X$, $Y$ operators, so it is a starter.\\
    \item $\text{YXIIZZIIYYII}$$-$starter, \\
    
    This Pauli string conserves the symmetries of the Hartree-Fock state and at the same time contains exactly four $X$, $Y$ operators, so it is a starter.\\
 	\item $\text{ZIXYZZZIYYII}$$-$starter, \\
 	
     This Pauli string conserves the symmetries of the Hartree-Fock state and at the same time contains exactly four $X$, $Y$ operators, so it is a starter.\\
	\item $\text{XXIZZZYXIIII}$$-$starter,\\
	
     This Pauli string conserves the symmetries of the Hartree-Fock state and at the same time contains exactly four $X$, $Y$ operators, so it is a starter.\\
	\item $\text{XYZIZIYYZIII}$$-$starter,\\
	
     This Pauli string conserves the symmetries of the Hartree-Fock state and at the same time contains exactly four $X$, $Y$ operators, so it is a starter.\\	
	\item $\text{IIYXYYZZZZII}$$-$starter,\\
	
     This Pauli string conserves the symmetries of the Hartree-Fock state and at the same time contains exactly four $X$, $Y$ operators, so it is a starter.\\	
	\item $\text{ZZYXIZYYIIII}$$-$starter,\\
	
     This Pauli string conserves the symmetries of the Hartree-Fock state and at the same time contains exactly four $X$, $Y$ operators, so it is a starter.\\	
	\item $\text{YZIXZZZIIYYI}$$-$starter,\\
	
     This Pauli string conserves the symmetries of the Hartree-Fock state and at the same time contains exactly four $X$, $Y$ operators, so it is a starter.\\	
	\item $\text{IXXZIIIZZXYI}$$-$starter,\\
	
     This Pauli string conserves the symmetries of the Hartree-Fock state and at the same time contains exactly four $X$, $Y$ operators, so it is a starter.\\	
	\item $\text{XIIIZZXYYIXY}$$-$not a starter,\\
	
  	 This Pauli does not preserve the number of particles in the Hartree-Fock state, so it cannot be a starter. But it obeys rules 1-2.\\
	\item $\text{XXXZYXXXYXYI}$$-$not a starter,\\
	
  	 This Pauli does not preserve the number of particles in the Hartree-Fock state, so it cannot be a starter. But it obeys rules 1-2.\\
	\item $\text{ZXIIIZZZZYII}$$-$not a starter,\\
	
  	 This Pauli is a single-particle excitation, so it cannot be a starter. But it obeys rules 1-2.\\
	\item $\text{XIZZIZXYZXII}$$-$not a starter,\\
	
  	 This Pauli does not preserve the number of particles in the Hartree-Fock state, so it cannot be a starter. But it obeys rules 1-2.\\
	\item $\text{YZXZZIZZYZYI}$$-$starter,\\
	
     This Pauli string conserves the symmetries of the Hartree-Fock state and at the same time contains exactly four $X$, $Y$ operators, so it is a starter.\\
	\item $\text{YZYXZIZIXZXY}$$-$not a starter,\\
	
	This Pauli contains more than four $X$, $Y$ operators, so it is not a starter. But it obeys rules 1-2.\\
	\item $\text{ZZZIZIIZXXXY}$$-$not a starter,\\
	
   This Pauli does not preserve the number of particles in the Hartree-Fock state, so it cannot be a starter. But it obeys rules 1-2.\\
	\item $\text{IZZZYYYXYXXY}$$-$not a starter,\\
	
   This Pauli does not preserve the number of particles in the Hartree-Fock state, so it cannot be a starter. But it obeys the rules 1-2.\\
\end{itemize}

\bibliographystyle{quantum}
\bibliography{references}
\end{document}